\documentclass[runningheads]{llncs}
\usepackage{booktabs}   
\usepackage{subcaption}
\usepackage{graphicx}
\usepackage{amsmath}
\usepackage{listings, framed}
\usepackage{amsfonts}
\usepackage{mathtools}
\usepackage{amssymb}
\usepackage{stmaryrd}
\usepackage{graphicx}
\usepackage{cleveref}
\usepackage{relsize}
\usepackage{bbm}
\usepackage{enumerate}
\usepackage{xcolor}
\usepackage{xspace}
\usepackage{amsmath}
\usepackage{adjustbox}
\usepackage{booktabs}
\usepackage{multirow}
\usepackage{array}
\usepackage{arydshln}
\usepackage{wrapfig}
\usepackage{proof}
\usepackage{algorithm}
\usepackage[noend]{algpseudocode}
\usepackage{nicefrac}
\usepackage[disable]{todonotes}

\newcommand{\opcit}{\textit{op.}\ \textit{cit.}\ }

\newcommand{\cf}{\textit{cf.}\  }
\newcommand{\ie}{i.e.\ }
\newcommand{\eg}{e.g.\ }
\newcommand{\wrt}{w.r.t.\ }

\crefname{section}{\S\hspace{-2pt}}{}
\crefname{subsection}{\S\hspace{-2pt}}{}
\crefname{subsubsection}{\S\hspace{-2pt}}{}
\crefname{appendix}{Appendix}{Appendices}
\crefname{equation}{Eq.}{Eqs.}
\crefname{figure}{Fig.}{}
\crefname{assumption}{Assumption}{}
\crefname{lemma}{Lemma}{}
\crefname{proposition}{Prop.}{}
\crefname{corollary}{Corollary}{}
\crefname{theorem}{Theorem}{}
\crefname{algorithm}{Algorithm}{Algorithms}

\newenvironment{myproof}[1]{\noindent\textbf{Proof of #1:}}{\hfill$\square$}
\newcommand{\mydashline}{\hdashline[1pt/1pt]}
\newcolumntype{R}[1]{>{\raggedleft\let\newline\\\arraybackslash\hspace{0pt}}m{#1}}
\newcolumntype{C}[1]{>{\centering\let\newline\\\arraybackslash\hspace{0pt}}m{#1}}
\newcolumntype{L}[1]{>{\raggedright\let\newline\\\arraybackslash\hspace{0pt}}m{#1}}

\newcommand{\Pro}[1]{\mathbb{P}\left[ #1 \right]}
\newcommand{\uniform}[2]{\mathsf{Unif}(#1,#2)}
\newcommand{\normal}[2]{\mathsf{Norm}(#1,#2)}

\newcommand{\cov}{\mathrm{Cov}}
\newcommand{\Exp}[1]{\mathbb{E}\left[#1\right]}

\newcommand{\uro}{u}
\newcommand{\emin}{e_{min}}
\newcommand{\emax}{e_{max}}
\newcommand{\round}{\mathrm{Round}}
\newcommand{\eabs}{\mathrm{err}_{\mathrm{abs}}}
\newcommand{\erel}{\mathrm{err}_{\mathrm{rel}}}
\newcommand{\ceil}[1]{\lceil #1 \rceil}
\newcommand{\floor}[1]{\lfloor #1 \rfloor}

\newcommand{\fintvl}[1][z]{\mathlarger{\lfloor}#1,#1\mathlarger{\rceil}}

\newcommand{\F}{\mathbb{F}}
\newcommand{\R}{\mathbb{R}}
\newcommand{\eR}{\overline{\R}}
\newcommand{\N}{\mathbb{N}}
\newcommand{\mop}{~\mathtt{op_m}~}
\newcommand{\iop}{~\mathrm{op}~}
\newcommand{\one}{\mathbbm{1}}
\newcommand{\norm}[1]{\lVert #1\rVert}
\newcommand{\ssvin}{\hspace{-2pt}\in\hspace{-2pt}} 

\newcommand{\absv}[1]{\left\lvert #1\right\rvert}
\newcommand{\dt}{\frac{\partial}{\partial t}}
\newcommand{\arctanh}{\mathrm{arctanh}}
\newcommand\defeq\triangleq
\newcommand{\sem}[1]{\left\llbracket#1\right\rrbracket}

\newcommand{\Tool}{\textsc{PAF}\xspace}

\newcommand{\bw}{{\bw w}}

\providecommand{\defeq}{\mathrel{:=}}

\begin{document}

\title{Rigorous Roundoff Error Analysis of Probabilistic Floating-Point
Computations\thanks{Supported in part by the National Science
Foundation awards CCF 1552975, 1704715,  the Engineering and Physical
Sciences Research Council (EP/P010040/1), and the Leverhulme Project Grant ``Verification of Machine Learning Algorithms''.}}

\titlerunning{Roundoff Error Analysis of Probabilistic Floating-Point
Computations}

\author{George Constantinides\inst{1} \and
Fredrik Dahlqvist\inst{1,2} \and
Zvonimir Rakamaric\inst{3}\and Rocco Salvia\inst{3}}

\authorrunning{G.~Constantinides \and
F.~Dahlqvist \and
Z.~Rakamaric \and R.~Salvia}

\institute{Imperial College London \email{g.constantinides@ic.ac.uk} \and
University College London \email{f.dahlqvist@ucl.ac.uk} \and University of Utah \email{\{rocco,zvonimir\}@cs.utah.edu}}

\maketitle              

\begin{abstract}
We present a detailed study of roundoff errors in probabilistic floating-point computations.
We derive closed-form expressions for the distribution of roundoff errors associated with a random variable, and we prove that roundoff errors are generally close to being uncorrelated with their generating distribution. Based on these theoretical advances, we propose a model of IEEE floating-point arithmetic for numerical expressions with probabilistic inputs and an algorithm for evaluating this model.  Our algorithm provides rigorous bounds to the output and error distributions of arithmetic expressions over random variables, evaluated in the presence of roundoff errors.  It keeps track of complex dependencies between random variables using an SMT solver,  and is capable of providing sound but tight probabilistic bounds to roundoff errors using symbolic affine arithmetic.  We implemented the algorithm in the PAF tool, and evaluated it on FPBench, a standard benchmark suite for the analysis of roundoff errors. Our evaluation shows that PAF computes tighter bounds than current state-of-the-art on almost all benchmarks.
\end{abstract}

\section{Introduction}
\label{sec:introduction}

There are two common sources of randomness in a numerical computation (a
straight-line program). First, the computation might be using inherently noisy data, for example from analog sensors in cyber-physical systems such as robots, autonomous vehicles, and
drones.  A prime example is data from GPS sensors, whose error distribution can be described very precisely~\cite{bornholt2014uncertain} and which we study in some detail in \cref{sec:overview}.
Second, the computation itself might sample from random number
generators. Such probabilistic numerical routines, known as Monte-Carlo
methods, are used in a wide variety of tasks, such as integration~\cite{press1990recursive,lepage1980vegas},
optimization~\cite{press1988numerical}, finance~\cite{glasserman2013monte}, fluid
dynamics~\cite{landau2014guide}, and computer graphics~\cite{kajiya1986rendering}.  We call
 numerical computations whose input values are sampled from
some probability distributions \emph{probabilistic computations}.

Probabilistic computations are typically implemented using
floating-point arithmetic, which leads to roundoff errors being introduced in
the computation. To strike the right balance between the performance
and energy consumption versus the quality of the computed result, expert
programmers rely on either a manual or automated floating-point error analysis
to guide their design decisions. However, the current state-of-the-art
approaches in this space have primary focused on \emph{worst-case} roundoff error
analysis of \emph{deterministic} computations.  So what can we say about
floating-point roundoff errors in a probabilistic context? Is it possible to
probabilistically quantify them by computing confidence intervals?  Can we, for
example, say with 99\% confidence that the roundoff error of the computed
result is smaller than some chosen constant?  What is the distribution of
outputs when roundoff errors are taken into account? In this paper, we explore these and similar questions.  To answer them, we propose a rigorous -- that is to say \emph{sound} -- approach to quantifying
roundoff errors in probabilistic computations. Based on this
approach, we develop an automatic tool that efficiently computes an overapproximate
probabilistic profile of roundoff errors.

As an example, consider the floating-point arithmetic expression $(X+Y)\div Y$, where $X$ and $Y$ are random inputs described formally as independent random variables (usually defined over some intervals). In \cref{sec:errordist}, we first show how the computation in \emph{finite-precision} of a single arithmetic operation such as $X+Y$ can be modelled as $(X+Y)(1+\varepsilon)$, where $\varepsilon$ is also a random variable. We then show how this random variable can be computed from first principles and why it makes sense to view $(X+Y)$ and  $(1+\varepsilon)$ as independent expressions, which in turn allows us to easily compute the distribution of $(X+Y)(1+\varepsilon)$.
The distribution of $\varepsilon$ depends on that of $X+Y$, and we therefore need to evaluate arithmetic operations between random variables. When the operands are independent -- as in $X+Y$ -- this is standard~\cite{springer1979algebra}, but when the operands are dependent -- as in the case of the division in $(X+Y)\div Y$ -- this is a hard problem. To solve it, we adopt and improve a technique for soundly bounding these distributions described in~\cite{bouissou2012generalization}.  Our improvement comes from the use of an SMT solver to reason about the dependency between $(X+Y)$ and $Y$ and remove regions of the state-space with zero probability. We describe this in \cref{sec:model}.

We can thus soundly bound the output distribution of any probabilistic computation, such as $(X+Y)\div Y$, performed in finite-precision floating-point arithmetic. This gives us the ability to perform \emph{probabilistic range analysis} and prove rigorous assertions like: 99\% of the outputs of a floating-point computation are smaller than a given constant bound.
In order to perform \emph{probabilistic roundoff error analysis} we develop \emph{symbolic affine arithmetic} in \cref{sec:symbolicaffine}. This technique is combined with probabilistic range analysis to compute \emph{conditional roundoff errors}. Specifically, we over-approximate the maximal error conditioned on the output landing in the 99\% range computed by the probabilistic range analysis, \ie conditioned on the computations not returning an outlier.  This allows us to quantify very precisely the idea of the maximal error of a `typical' computation.

Here we need to make a comment about how to understand the notion of \emph{sound} or \emph{rigorous} error bounds in a probabilistic context.  The usual meaning of these words is anchored in a non-deterministic model of computation where inputs are known to belong to some intervals and soundly bounding the error amounts to bounding the error of \emph{all possible computations}.
If the inputs are drawn from a distribution, this understanding of soundness corresponds to bounding errors with 100\% confidence, and thereby ignoring the probabilistic information completely.
This is why the state-of-the-art refers to this non-deterministic model as \emph{worst-case} analysis, because it does not account for the distribution of the inputs.
On the other hand, in a probabilistic model of computation,  we will say that we are giving a sound (or rigorous) probabilistic guarantee if there is 100\% certainty about the probabilistic guarantee.   For example when saying \emph{99\% of outputs will land in the interval $[a,b]$} we are making a sound probabilistic statement based on some analytical description of the output distribution, and there is therefore no uncertainty about this statement.  An example of \emph{unsound} probabilistic statement is given by making probabilistic statements based on Monte-Carlo sampling. When saying \emph{based on 1 billion runs, 99\% of outputs will land in the interval $[a,b]$} we are making an unsound probabilistic statement which may fail in another run of the experiment. Another example of unsound probabilistic statement would be \emph{with 95\% certainly, 99\% of outputs will land in the interval $[a,b]$}, as is done in \eg PAC learning \cite{valiant1984theory}.  By working with `p-boxes'~\cite{bouissou2012generalization} -- structures which upper- and lower-bound cumulative distribution functions (see \cref{sec:preliminaries} and \cref{sec:errordist}) -- we will be able to give probabilistic guarantees which are all \emph{sound}.

We implemented our model and algorithms in a tool called \Tool (for Probabilistic Analysis of Floating-point errors). We evaluated \Tool on the standard floating-point benchmark suite FPBench~\cite{fpbench}, and compared its range and error analysis with the worst-case roundoff error analyzer FPTaylor~\cite{2015_fm_sjrg,solovyev2018rigorous} and the probabilistic roundoff error analyzer PrAn~\cite{probdaisy}. We present the results in \cref{sec:evaluation},  and show that FPTaylor's worst-case analysis is often overly pessimistic in the probabilistic setting, while \Tool also generates tighter probabilistic error bounds than PrAn on all but one benchmark. 

We summarize our contributions as follows:
\begin{enumerate}[(i)]
\item We derive several closed-form expressions for the distribution of roundoff errors associated with a random variable.  We prove that roundoff errors are generally close to being uncorrelated with their input distribution.
\item Based on these results we propose a model of IEEE754 floating-point arithmetic for numerical expressions with probabilistic inputs.
\item We evaluate this model by developing a new algorithm for rigorously bounding the output range and roundoff error distributions of floating-point arithmetic expressions with probabilistic inputs.
\item We implement this model in the \Tool tool, and perform probabilistic range and roundoff error analysis on a standard benchmark suite. Our comparison
	with the current state-of-the-art shows the advantages of our approach in terms of computing tighter, and yet still rigorous,
	probabilistic bounds.
\end{enumerate}

\section{Overview Through an Application}\label{sec:overview}

GPS sensors provide latitude and longitude  coordinates that are inherently noisy.  As shown in detail in~\cite{bornholt2013abstractions}, the conditional probability distribution of the true longitude and latitude coordinates $(\mathtt{TrueLat}, \mathtt{TrueLong})$,  given a GPS sensor reading $(\mathtt{ReadLat}, \mathtt{ReadLong})$, is distributed according to a Rayleigh distribution
\begin{align*}
 & \Pro{\mathrm{Location}=(\mathtt{TrueLat}, \mathtt{TrueLong}) \mid \mathrm{GPS}=(\mathtt{ReadLat}, \mathtt{ReadLong}) } \\
\sim~ & \mathrm{Rayleigh}\left(\norm{\mathtt{TrueLat}, \mathtt{TrueLong}) - (\mathtt{ReadLat}, \mathtt{ReadLong}) }; \nicefrac{\varepsilon}{\sqrt{\ln 400}}\right)
\end{align*}
where $\varepsilon$ is a hardware dependent measure of accuracy of the GPS reading known as horizontal accuracy.  Interestingly and somewhat surprisingly, since the density of any Rayleigh distribution is always zero at $x=0$,  it is extremely unlikely that the true coordinates lie in a small neighbourhood of those given by the GPS reading.  This leads to errors described in \cite{bornholt2013abstractions} and \cite{bornholt2014uncertain} and to the suggestion that the coordinates given by a GPS sensor should be corrected by adding a probabilistic error term which will, on average, shift the observed coordinates into an area of high probability for the true coordinates. \cref{fig:latitude} gives the correction of \cite{bornholt2014uncertain}, where the value of the variable $\mathtt{radius}$ is drawn from the Rayleigh distribution with parameter $\frac{\varepsilon}{\sqrt{\ln 400}}$, the value of the variable $\mathtt{angle}$ is draw from the uniform distribution on $[0,2\pi]$, and  \texttt{DEGREES\_PER\_METER} is a constant ($\approx$ 0.000009009) used to convert the $\mathtt{radius}$ variable (expressed in meters) into degrees.
\begin{figure}[b]
	\begin{lstlisting}
	TrueLat = ReadLat + ((radius * sin(angle)) * DEGREES_PER_METER)
	TrueLong = ReadLong + ((radius * cos(angle)) * DEGREES_PER_METER)
	\end{lstlisting}
	\caption{Probabilistic correction of GPS coordinates from \cite{bornholt2014uncertain}: $\mathtt{radius}$ is Rayleigh distributed, $\mathtt{angle}$ is uniformly distributed.}\label{fig:latitude}
\end{figure}

A developer trying to strike the right balance between resources, such as energy consumption or execution time, versus the accuracy of the computation, particularly in the context of mobile and embedded devices, might want to run a rigorous worst-case floating-point analysis tool to decide which floating-point format is accurate enough to process GPS signals.
This is a mandatory decision if the developer requires rigorous error bounds holding with 100\% certainty.
The problem such a developer would face when analyzing a piece of code involving the probabilistic correction above, is that the Rayleigh distribution --- that is to say the value of the variable $\mathtt{radius}$ in \cref{fig:latitude} --- has a semi-infinite support (\ie $[0, \infty)$), and \emph{any} rigorous (\ie sound) worst-case roundoff error analysis would return an infinite error bound in the presence of unbounded variables.
This is because, in the \emph{worst-case} scenario, the unbounded variable assumes exactly the infinite value, thus resulting in an infinite roundoff error.
In order to get a meaningful (numeric) error bound from worst-case roundoff error tools, the range of $\mathtt{radius}$ will therefore have to be truncated. The main consequence of truncation is that soundness --- which is the purpose of these tools --- is lost, because regardless of \emph{where} you decide to truncate the distribution, you are going to discard a part of its support with non-zero probability. To mitigate this issue and be `as sound as possible', the natural truncation is given by the interval $[0,max_{(p, e)}]$ where $max_{(p, e)}$ is the largest representable number not causing an overflow, in a generic floating-point format with $p$ bits for the mantissa and $e$ bits for the exponent representations.
As we show next, the main issue with this natural truncation is that all the useful probabilistic information is lost, and regions of high probability are treated in the same way as regions which would never be explored even if we sampled for billions of years.

As an example, let us consider the latitude correction of \cref{fig:latitude}:
\begin{align}
\mathtt{ReadLat + ((radius * sin(angle)) * DEGREES\_PER\_METER)},\label{latitude}
\end{align}
and set $\mathtt{ReadLat}$ to $51.4769^{\circ}$, the latitude of the Greenwich observatory, and DEGREES\_PER\_METER to 0.000009009.
In \cref{tab:erroranalysisGPS}, we report our detailed roundoff error analysis when \eqref{latitude} is implemented in the well-known IEEE-754 double-, single-, and half-precision formats.
With each implementation format (reported in the Format column), we associate the range $[0,max_{(p,e)}]$, with $max_{(p,e)}$ reported in the Max column, of the Rayleigh distribution truncated in the natural manner described above. In double-precision, for example, we truncate the Rayleigh distribution to the range $[0, 10^{307}]$ since $10^{307}\approx max_{(53,11)}$.
We computed worst-case roundoff error bounds for \eqref{latitude}  in each precision format with the state-of-the-art error analyzer FPTaylor~\cite{2015_fm_sjrg}, with the variable \texttt{radius} constrained to the corresponding truncated range. We show the results in the column FPTaylor.
Note that in double-precision, FPTaylor returns a roundoff error bound of 4.3e+286. 
The reason for such an enormous error is the fact that FPTaylor does not take into consideration how \texttt{radius} is distributed in the range, \ie how rare the values contributing to such a large error are.
We also computed worst-case roundoff errors with our tool \Tool by setting the confidence interval from which conditional roundoff errors are computed to 100\%. We report the results in the PAF 100\% column; these results are identical to the ones from FPTaylor,  thereby confirming empirically that worst-case roundoff error is synonymous with conditional roundoff error based on a 100\% confidence interval.
Finally, the PAF 99.9999\% column gives the 99.9999\% \emph{conditional roundoff error} computed by \Tool. This value is an upper bound to the roundoff error \emph{conditioned} on the computation having landed in an interval capturing (at least) 99.9999\% of all possible outputs.  The column Absolute gives the 99.9999\% conditional roundoff error of the latitude correction \eqref{latitude}, which is expressed in degrees, and the column Meters converts this value into meters, based on the fact that $1^{\circ}$ of latitude is roughly equal to 111km.

\begin{table*}[tb]
	\centering
	\caption{{Roundoff error analysis for the probabilistic latitude correction of \eqref{latitude}.}}
\label{tab:erroranalysisGPS}
	\renewcommand{\arraystretch}{1.1}
	\begin{tabular}{@{\extracolsep{3pt}}p{1.5cm}L{1.4cm}L{1.6cm}L{1.8cm}L{1.6cm}L{1cm}@{}}
		\toprule
		\multicolumn{4}{c}{} &
		 \multicolumn{2}{c}{\Tool 99.9999\%}\\
		 \cmidrule{5-6}
		 Precision & Max & FPTaylor & \Tool 100\% & Absolute & Meters\\
		\midrule
		double & $\approx 10^{307}$ & 4.3e+286 & 4.3e+286 & 4.1e-15 & 4.5e-10 \\
		\mydashline{}
		single & $\approx 10^{38}$ & 2.1e+26 & 2.1e+26 & 3.7e-06 & 4.1e-1 \\
		\mydashline{}
		half & $\approx 10^{4}$ & 2.5e-2 & 2.5e-2 & 2.4e-2 & 2667 \\
		\bottomrule
	\end{tabular}
\end{table*}

Our probabilistic error analysis in \Tool with 99.9999\% confidence interval reduces the 100\% (i.e., worst-case) error bound by 300 orders of magnitude in the case of double-precision, and by 30 orders of magnitude in the case of single-precision.
The reason behind such dramatic reduction is that \Tool runs a sound probabilistic range analysis (described in \cref{sec:errordist} and \cref{subsec:dependentOperation}) in order to over-approximate the 99.9999\% confidence interval when \eqref{latitude} is implemented in floating-point arithmetic. This 99.9999\% interval is then used  to compute the \emph{conditional roundoff error} (described in \cref{subsec:conderror}).
Specifically, using \emph{symbolic affine arithmetic} (see \cref{sec:symbolicaffine}), \Tool will maximize the possible roundoff error \emph{for those computations that land in the 99.9999\% confidence interval}.
In other words, \Tool over-approximates the roundoff error of all possible computations except the 0.0001\% of outliers that lead to the most abnormal outputs.

Without our probabilistic error analysis, the developer might \emph{erroneously} conclude that half-precision format is the most appropriate to implement \eqref{latitude} because it results in the smallest error bound.
However, with the information provided by the 99.9999\% conditional roundoff error, the developer can see that the \emph{average} error, without the extremely rare outliers, is many orders of magnitude smaller than the worst-case scenarios.
Furthermore, the developer can also observe that in the case of the half-precision format (and in this case alone), the worst-case errors and the probabilistic errors are almost identical,  which means that the \emph{average} roundoff error is of the same order of magnitude as the worst-case error.
This is because in half-precision (and lower), the roundoff error of \eqref{latitude} is governed by the rounding errors of the constants,  and is thus completely independent from the chosen confidence interval.

Armed with this information, the developer can conclude that with a roundoff error of roughly 40cm (4.1e-1 meters) when correcting 99.9999\% of GPS latitude readings, working in single-precision is an adequate compromise between efficiency and accuracy of the computation, since double-precision provides unnecessary accuracy (roundoff error of 4.5e-10 meters) and half-precision is \emph{always} overly imprecise (roundoff error of 2667 meters).  As a quantitative intuition of the 99.9999\% confidence level,  it is not hard to show (by using the Poisson approximation of the binomial distribution) that the probability of encountering at least one outlier after 700,000 latitude corrections given by \eqref{latitude} is about 50\%. Assuming one GPS reading every 5 seconds, this corresponds to at least 50\% chance that the roundoff error bound is never breached after 40 days of continual use.  Of course, we can run \Tool at higher levels of confidence if required.

This demonstrates the innovative concept of \emph{probabilistic precision tuning}, supported by this paper, in order to determine which floating-point format is the most appropriate to implement \eqref{latitude}.
In the current worst-case precision tuners~\cite{fptuner,darulova2018daisy}, the developer provides the tuner with an algebraic expression $f(x)$ (like \eqref{latitude}) where the arguments $x$ are properly bounded, together with a numerical roundoff error bound $\epsilon$.
The output from the tuner is the minimal (in terms of bits) floating-point format required when implementing the expression $f(x)$ that guarantees that the implementation will \emph{never} violate the given error bound.
Clearly, worst-case precision tuning suffers from the same limitations we discussed above, meaning we have no information on how rare the inputs violating the given error bound might be. 

As an example, let us do a precision tuning exercise for the latitude correction computation \eqref{latitude}. We bound the variable $\mathtt{radius}$ in the interval $[0, 10^{307}]$ and we set the error bound $\epsilon$ to 1e-5 (roughly $1m$).
We manually perform worst-case precision tuning using FPTaylor to determine that the minimal floating-point format not violating the given error bound $\epsilon$ consists of $1022$ bits for the mantissa and $11$ bits for the exponent.
Such custom floating-point format is prohibitively expensive, in particular for devices with an intensive usage of GPS readings, like smart-phones or smart-watches. 
Thus, worst-case error bounds are not only very pessimistic, but are also of little practical use from an implementation prospective in some cases.

%
On the other hand, when we manually performed \emph{probabilistic precision tuning} using \Tool with a confidence interval set to 99.9999\%, we determine that the minimal required floating-point format consists of $22$ bits for the mantissa and $11$ bits for the exponent.
Thanks to \Tool, the confidence interval has now become an additional input parameter of the probabilistic precision tuning routine. Thus, in addition to the expression $f(x)$ and the error bound $\epsilon$, the developer can also provide a custom confidence interval of interest and ignore the extremaly unlikely corner cases like the ones we described for \eqref{latitude}.

\section{Preliminaries}\label{sec:preliminaries}

\subsection{Floating-Point Arithmetic}\label{subsec:fparith}

Given a \emph{precision} $p\in\N$ and an \emph{exponent range} $[\emin,\emax]\defeq \{ n \mid n \in \N \wedge \emin \leq n \leq \emax \}$, we define $\F(p,\emin,\emax)$, or simply $\F$ if there is no ambiguity, as the set of extended real numbers
\begin{align*}
\F\defeq\left\{\left.(-1)^s 2^e\left(1+\frac{k}{2^p}\right)\right\rvert s\in\{0,1\}, e\in [\emin,\emax], 0\leq k<2^p\right\}\cup\{-\infty,0,\infty\}
\end{align*}
Elements $z=z(s,e,k)\in \F$ will be called \emph{floating-point representable numbers} (for the given precision $p$ and exponent range $[\emin,\emax]$) and we will use the variable $z$ to represent them. The variable $s$ will be called the \emph{sign}, the variable $e$ the \emph{exponent}, and the variable $k$ the \emph{significand} of $z(s,e,k)$.
  
Next, we introduce a \emph{rounding map} $\round:\R\to\F$ that rounds to nearest (or to $-\infty$/$\infty$ for values smaller/greater than the smallest/largest finite element of $\F$) and follows any of the IEEE754 rounding modes in case of a tie. We will not worry about which choice is made since the set of mid-points will always have probability zero for the  distributions we will be working with. All choices are thus equivalent, probabilistically speaking, and what happens in a tie can therefore be left unspecified.  We will denote the extended real line by $\eR\defeq \R\cup\{-\infty,\infty\}$. The (signed) \emph{absolute error function} $\eabs:\R\to\eR$ is defined as:
$\eabs(x)=x-\round(x)$.
We define the sets $\fintvl\stackrel{\triangle}{=}\left\{y\in\R\mid \round(y)=\round(z)\right\}$.   
Thus if $z\in \F$, then $\fintvl[z]$ is the collection of all reals rounding to $z$.  As the reader will see, the basic result of \cref{sec:errordist} (\cref{eq:LPerrorDensity}) is expressed entirely using the notation $\fintvl[z]$ which is parametric in the choice of the $\round$ function.  It follows that our results apply to rounding modes other that round-to-nearest with minimal changes.  The \emph{relative error function} $\erel:\R\setminus\{0\}\to\eR$ is defined by
\[
\erel(x)=\frac{x-\round(x)}{x}.
\]
Note that $\erel(x)=1$ on $\fintvl[0]\setminus\{0\}$, $\erel(x)=\infty$ on $\fintvl[-\infty]$ and $\erel(x)=-\infty$ on $\fintvl[\infty]$.
Recall also the fact \cite{higham2002accuracy} that $-2^{-(p+1)}<\erel(x)<2^{-(p+1)}$ outside of $\fintvl[0]\cup\fintvl[-\infty]\cup\fintvl[\infty]$. The quantity $2^{-(p+1)}$ is usually called the \emph{unit roundoff} and will be denoted by $\uro$. 

For $z_1,z_2\in\F$ and $\mathrm{op}\in\{+,-,\times,\div\}$ an (infinite-precision) arithmetic operation, the traditional model of
IEEE754 floating-point arithmetic~\cite{ieee754} \cite{higham2002accuracy} states that the finite-precision implementation $\mathtt{op_m}$ of $\mathrm{op}$ must satisfy
\begin{align}
z_1\mop z_2=(z_1\iop z_2)(1+\delta) \qquad\absv{\delta}\leq u, \label{eq:traditional}
\end{align}
We leave dealing with subnormal floating-point numbers to future work. The model given by \cref{eq:traditional} stipulates that the implementation of an arithmetic operation can induce a relative error of magnitude \emph{at most} $u$. The exact size of the error is, however, not specified and \cref{eq:traditional} is therefore a \emph{non-deterministic model of computation}. It follows that numerical analyses based on \cref{eq:traditional} must consider \emph{all} possible relative errors $\delta$ and are fundamentally \emph{worst-case} analyses. 
Since the output of such a program might be the input of another,  one should also consider non-deterministic inputs, and this is indeed what happens with automated tools for roundoff error analysis, such as Daisy~\cite{darulova2018daisy} or FPTaylor~\cite{2015_fm_sjrg,solovyev2018rigorous}, which require for each variable of the program a (bounded) range of possible values in order to perform a worst-case analysis (\cf GPS example in \cref{sec:overview}).

In this paper, we study a model formally similar to \cref{eq:traditional}, namely
\begin{align}
z_1\mop z_2=(z_1\iop z_2)(1+\delta) \qquad\delta\sim dist.\label{eq:probabilistic}
\end{align}
The difference is that $\delta$ is now \emph{distributed according to} $dist$, a probability distribution whose support is $\left[-u,u\right]$.  In other words, we move from a non-deterministic to a \emph{probabilistic} model of roundoff errors.  This is similar to the `Monte Carlo arithmetic' of \cite{parker2000monte}, but whilst \opcit \emph{postulates} that $dist$ is the uniform distribution on $[-u,u]$, we compute $dist$ from first principles in \cref{sec:errordist}.

\subsection{Probability Theory}\label{subsec:prob}

To fix the notation and be self-contained, we present some basic notions of probability theory which are essential to what follows.  

\noindent \textbf{Cumulative Distribution Functions and Probability Density Functions.}
We assume that the reader is (at least intuitively) familiar with the notion of a (real) random variable. Given a random variable $X$ we define its Cumulative Distribution Function (CDF) as the function $c(t)\defeq\Pro{X\leq t}$. If there exists a non-negative integrable function $d: \R\to\R$ such that
\begin{align*}
c(t)\defeq\Pro{X\leq t}=\int_{-\infty}^t d(t)~dt
\end{align*}
then we call $d(t)$ the Probability Density Function (PDF) of $X$.  If it exists,  then it can be recovered from the CDF by differentiation $d(t)=\dt c(t)$  by the fundamental theorem of calculus.  

Not all random variables have a PDF: consider the random variable which takes value $0$ with probability $\nicefrac{1}{2}$ and value $1$ with probability $\nicefrac{1}{2}$. For this random variable it is impossible to write $\Pro{X\leq t}=\int d(t)~dt$. Instead, we will write the distribution of such a variable using the so-called Dirac delta measure at 0 and 1 as $\nicefrac{1}{2}\delta_0 + \nicefrac{1}{2}\delta_1$.  It is possible for a random variable to have a PDF covering part of its distribution -- this will be its \emph{continuous part} -- and a sum of Dirac deltas covering the rest of its distribution -- this will be its \emph{discrete part}. We will encounter examples of such random variables in \cref{sec:errordist}. Finally, note that if $X$ is a random variable and $f: \R\to\R$ is a (measurable) function, then $f(X)$ is a random variable.  In particular $\erel(X)$ is a random variable, which we will describe in  detail in \cref{sec:errordist}.

\noindent \textbf{Arithmetic on Random Variables.}
 Suppose $X,Y$ are \emph{independent} random variables with PDFs $f_X$ and $f_Y$, respectively.  
Using the arithmetic operations we can form new random variables $X+Y, X-Y, X\times Y, X\div Y$. The PDFs of these new random variables can be expressed as operations on $f_X$ and $f_Y$~\cite{springer1979algebra}:
\begin{align}
&f_{X+Y}(t)\defeq\hspace{-1pt}\int_{-\infty}^{\infty} \hspace{-3pt}f_X(x)f_Y(t-x)~dx &
f_{X-Y}(t)\defeq\hspace{-1pt}\int_{-\infty}^{\infty} \hspace{-3pt}f_X(x)f_Y(x-t)~dx\nonumber \\
&f_{X\times Y}(t)\defeq\hspace{-1pt}\int_{-\infty}^{\infty}\hspace{-1pt} \frac{1}{\absv{x}}f_X(x)f_Y\hspace{-2pt}\left(\frac{t}{x}\right)dx &
f_{X\div Y}(t)\defeq\hspace{-1pt}\int_{-\infty}^{\infty} \hspace{-3pt}\absv{x}f_X(x)f_Y(tx)dx\label{eq:pdfarithm}
\end{align}
It is important to note that these formulas are only valid if $X$ and $Y$ are assumed to be independent.
When an arithmetic expression containing variable repetitions is given a random variable interpretation, this independence can no longer be assumed.  In the expression $(X+Y)\div Y$ the sub-term $(X+Y)$ can be interpreted by using Eqs.  \eqref{eq:pdfarithm} if $X,Y$ are independent. However, the sub-terms $X+Y$ and $Y$ clearly cannot be interpreted as independent random variables, and the division operation can therefore not be evaluated using Eqs.  \eqref{eq:pdfarithm}. 

\noindent\textbf{Soundly Bounding Probabilities.} 
The constraint that the distribution of a random variable must integrate to 1 makes it impossible to order random variables in the `natural' way: if $\Pro{X\in A}\leq \Pro{Y\in A}$, then $\Pro{Y\in A^c}\leq \Pro{X\in A^c}$, \ie we cannot say that $X\leq Y$ if $\Pro{X\in A}\leq \Pro{Y\in A}$.
This means that we cannot 
quantify our probabilistic uncertainty about a random variable by sandwiching it between two other random variables as one would do with reals or real-valued functions.

One solution is to restrict the sets used in the comparison,  \ie declare that $X\leq Y$ iff $\Pro{X\in A}\leq \Pro{Y\in A}$ for $A$ ranging over a given set of `test subsets'.  Such an order can be defined by taking the `test subsets' to be all the intervals $(-\infty, x]$~\cite{rothschild1970increasing}. This order is now known as the \emph{stochastic order}. 
It follows from the definition of the CDF that this order can most elegantly be defined by saying that $X\leq Y$ iff $c_X\leq c_Y$, where $c_X$ and $c_Y$ are the CDFs of $X$ and $Y$, respectively.  If it is possible to sandwich an unknown random variable $X$ between known lower and upper bounds $X_{lower}\leq X\leq X_{upper}$ using the stochastic order then it becomes possible to give sound bounds to the quantities $\Pro{X\in[a,b]}$ via
\begin{align*}
\Pro{X\in[a,b]} = c_X(b) - c_X(a)\leq c_{X_{upper}}(b) - c_{X_{lower}}(a)
\end{align*}
\noindent\textbf{P-Boxes and DS-Structures.} As mentioned above, giving a random variable interpretation to an arithmetic expression containing variable repetitions cannot be done using Eqs. \eqref{eq:pdfarithm}. In fact, these interpretations are in general analytically intractable. Hence, a common approach is to give up on soundness and approximate such distributions using Monte-Carlo simulations.  We use this approach in our experiments to assess the quality of our sound results.
However,  we will also provide sound under- and over-approximations of the distribution of arithmetic expressions over random variables using the stochastic order discussed above.  Since $X_{lower}\leq X\leq X_{upper}$ is equivalent to saying that $c_{X_{lower}}(x)\leq c_X(x)\leq c_{X_{upper}}(x)$, the fundamental approximating structure will be a pair of CDFs satisfying $c_1(x)\leq c_2(x)$.  Such a structure is known in the literature as a \emph{p-box}~\cite{ferson2015constructing}, and has already been used in the context of probabilistic roundoff errors in related work~\cite{bouissou2012generalization,probdaisy}.  
The data of a p-box is equivalent to a pair of sandwiching distributions for the stochastic order.

A \emph{Dempster-Shafer structure} (DS-structure) of size $N$ is a set of interval-probability pairs \{$([x_{0}, y_{0}], p_{0}), ([x_{1}, y_{2}], p_{1}),.., ([x_{N}, y_{N}], p_{N})$\} where $\sum_{i=0}^{N}p_{i}=1$.  The intervals in the collection might overlap. One can always convert a DS-structure to a p-box and back again~\cite{ferson2015constructing},  but arithmetic operations are much easier to perform on DS-structures than on p-boxes (\cite{bouissou2012generalization}),
which is why we will use DS-structures in the algorithm described in \cref{sec:model}.

\section{Distribution of Floating-Point Roundoff Errors}\label{sec:errordist}

We briefly described in the GPS example of \cref{sec:overview} how our tool \Tool computes \emph{probabilistic} roundoff errors by conditioning the optimization of symbolic affine form (presented in the next section) on the output of the computation landing in a confidence interval.  The purpose of this section is to provide the necessary probabilistic tools to compute these intervals. Specifically, we will show how the probabilistic model of \cref{eq:probabilistic} can be used to compute (or at least bound) the output distribution of a probabilistic computation in finite-precision. In other words, this section provides the foundations of \emph{probabilistic range analysis}.

Our first task is to specify the distribution $dist$ of the random variable $\delta$ in the probabilistic model given by \cref{eq:probabilistic}.  Following \cite{dahlqvist2019probabilistic},  we use the fact that $dist$ can be computed explicitly from first principles and expressed in closed form (\cref{eq:LPerrorDensity}).  However,  this expression can only be computed in practice at very low precision. We then show that roundoff error distributions in high precision (say, single-precision and higher) can be approximated by another closed-form expression \cref{eq:HPerrorDensity} with an error that can be explicitly bounded, allowing us to remain sound.  Moreover, the quality of this approximation increases with the working (\ie benchmark) precision.  In the case where all mantissas are equiprobable,  it can be shown that as the working precision increases,  the roundoff error distribution always converges to a unique distribution given by \cref{eq:typicalpdf}, which we call the \emph{typical distribution}.  
For this class of input distributions, the distribution of errors is therefore asymptotically independent of the way the inputs are distributed.  More generally, we show that the covariance between the roundoff errors of \cref{eq:HPerrorDensity} and their generating input distribution is extremely small, that is to say the two processes are almost completely uncorrelated. This will have an important practical application when evaluating our probabilistic model of IEEE 754 arithmetic presented in \cref{subsec:model}. All proofs can be found in the Appendix.

\subsection{Derivation of the Distribution of Rounding Errors}\label{subsec:LPerror_dist}

Let us return to the probabilistic model of IEEE 754 arithmetic given by \cref{eq:probabilistic} where $\iop$ is the infinite-precision arithmetic operation and $\mop$ is its finite-precision implementation:
\[
z_1\mop z_2=(z_1\iop z_2)(1+\delta) \qquad\delta\sim dist.
\]
Let us also assume that $z_1,z_2$ are random variables with known distributions.  Then $z_1\iop z_2$ is also a random variable which can (in principle) be computed. Since the IEEE 754 standard states that $z_1\mop z_2$ should be given by rounding the infinite precision operation $z_1\iop z_2$, it is a completely natural consequence of the standard to require that $\delta$ should simply be given by 
\[
\delta = \erel(z_1\iop z_2).
\]
Thus $dist$ should simply be the distribution of the random variable $\erel(z_1\iop z_2)$. 
More generally, if $X$ is a random variable with know distribution, we will show how to compute the distribution $dist$ of the random variable
\[
\erel(X)=\frac{X-\round(X)}{X}.
\]
We choose to express the distribution $dist$ of relative errors \emph{in multiples of the unit roundoff $\uro$}. This choice is arbitrary, but it allows us to work with a distribution on the conceptually and numerically convenient interval $\left[-1,1\right]$, since the absolute value of the relative error is strictly bounded by $\uro$ (see \cref{subsec:fparith}), rather than the interval $\left[-\uro,\uro\right]$. 

To compute the density function of $dist$, we proceed as described in \cref{subsec:prob} by first computing the CDF $c(t)$ and then taking its derivative.  Recall first from \cref{subsec:fparith} that $\erel(x)=1$ if $x\in \fintvl[0]\setminus\{0\}$,  $\erel(x)=\infty$ if $x\in \fintvl[-\infty]$,  $\erel(x)=-\infty$ if $x\in \fintvl[\infty]$, and $-u \leq \erel(x)\leq u$ elsewhere.  Thus:
\begin{align*}
& \Pro{\erel(X)=-\infty}=\Pro{X\in\fintvl[\infty]} 
& \Pro{\erel(X)=1}=\Pro{X\in\fintvl[0]} \\
& \Pro{\erel(X)=\infty}=\Pro{X\in\fintvl[-\infty]}
\end{align*}
In other words, the probability measure corresponding to $\erel$ has three discrete components at $\{-\infty\}, \{1\}$,  and $\{\infty\}$, which cannot be accounted for by a PDF (see \cref{subsec:prob}).  It follows that the probability measure $dist$ is given by
\begin{align}
dist_c + \Pro{X\ssvin\fintvl[0]}\delta_1 + \Pro{X\ssvin\fintvl[-\infty]}\delta_{\infty}\ \hspace{-3pt}+ \Pro{X\ssvin\fintvl[\infty]}\delta_{-\infty}\label{eq:densitydecomp}
\end{align}
where $dist_c$ is a continuous measure that is not quite a probability measure since its total mass is $1-\Pro{X\in\fintvl[0]}-\Pro{X\ssvin\fintvl[-\infty]} - \Pro{X\ssvin\fintvl[\infty]}$. 
In general, $dist_c$ integrates to 1 in machine precision since $\Pro{X\in\fintvl[0]}$ is of the order of the smallest positive floating-point representable number, and the PDF of $X$ rounds to 0 way before it reaches the smallest/largest floating-point representable number. For example, Python's \texttt{scipy} implementation of the PDF of the standard Gaussian rounds to 0 from $x=39$.  However in order to be sound, we must in general include these three discrete components in our computations. The density $dist_c$ can be computed explicitly and is given by the following result whose proof can essentially be found in \cite{dahlqvist2019probabilistic}.

\begin{theorem}\label{thm:LP_errorDensity}
Let $X$ be a real random variable with PDF $f$. The continuous part $dist_c$ of the distribution of $\erel(X)$ has a PDF given by
\begin{align}
d(t)=\sum_{z\in\F\setminus\{-\infty,0,\infty\}}\one_{\fintvl[z]}\left(\frac{z}{1-t\uro}\right) f\left(\frac{z}{1-t\uro}\right) \frac{\uro\absv{z}}{(1-t\uro)^2},\label{eq:LPerrorDensity}
\end{align}
where $\one_{A}(x)$ is the indicator function which returns 1 if $x\in A$ and 0 otherwise. 
\end{theorem}

\begin{figure}[tb]
\centering
\begin{tabular}{l l l}
\includegraphics[scale=0.25]{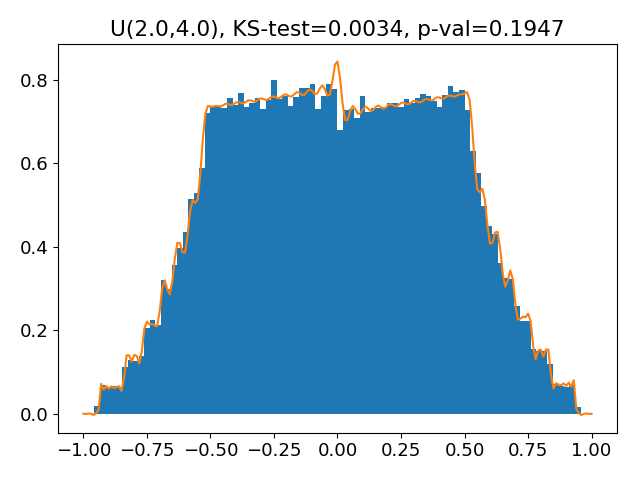}
&
\includegraphics[scale=0.25]{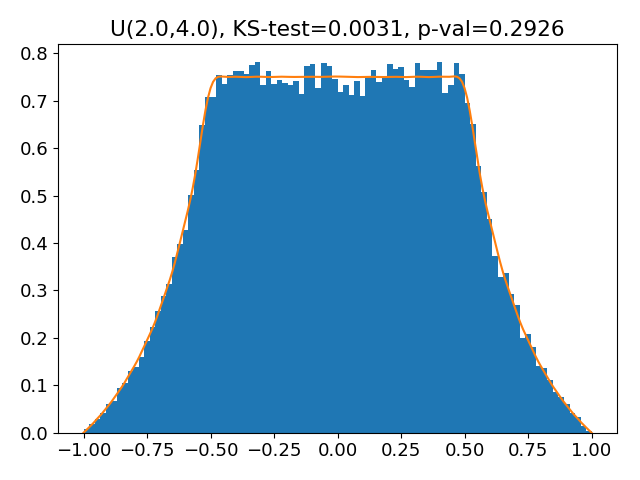}
&
\includegraphics[scale=0.25]{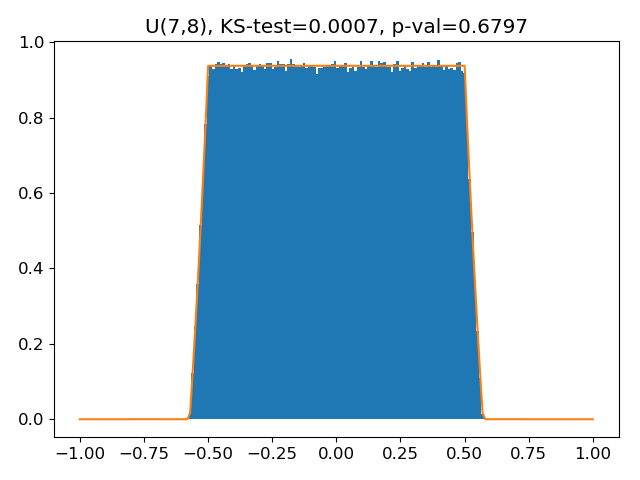}
\\
\includegraphics[scale=0.25]{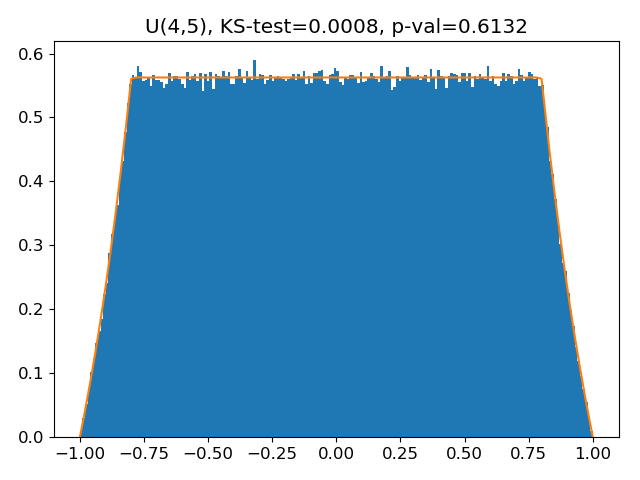}
&
\includegraphics[scale=0.25]{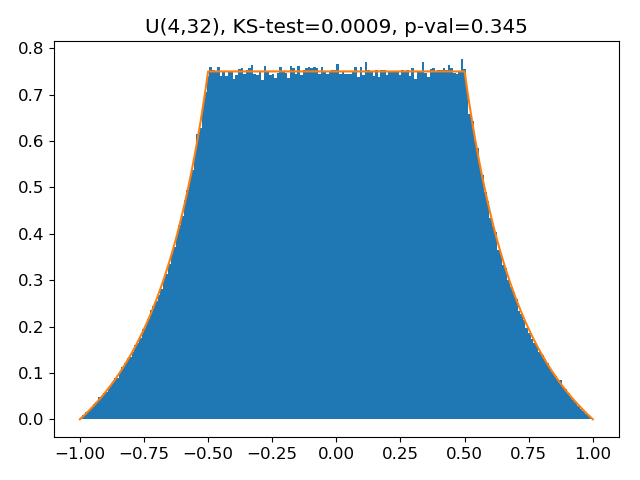}
&
\includegraphics[scale=0.25]{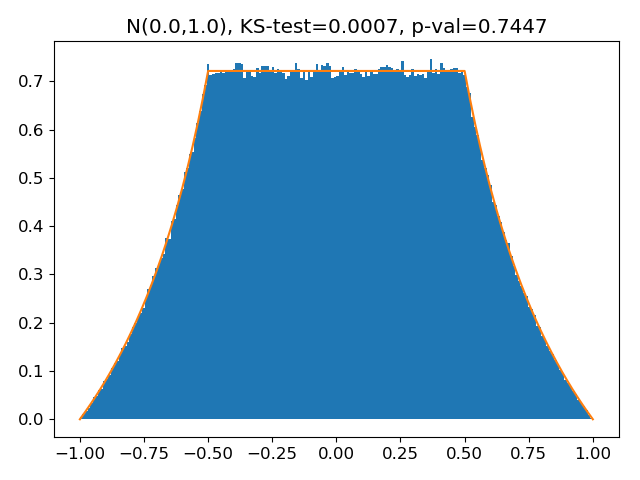}
\end{tabular}
\caption{Error distribution versus the corresponding empirical error distributions, clockwise from top-left: (i) \cref{eq:LPerrorDensity} for $\uniform{2}{4}$ 3 bit exponent,  4 bit significand,  (ii) \cref{eq:LPerrorDensity} for $\uniform{2}{4}$ in half-precision,  (iii) \cref{eq:HPerrorDensity} for $\uniform{7}{8}$ in single-precision, (iv) \cref{eq:HPerrorDensity} for $\uniform{4}{5}$ in single-precision, (v) \cref{eq:HPerrorDensity} for $\uniform{4}{32}$ in single-precision, (vi) \cref{eq:HPerrorDensity} for $\normal{0}{1}$ in single-precision.}
\label{fig:error_dist}
\end{figure}

\cref{fig:error_dist} (i) and (ii) shows an implementation of \cref{eq:LPerrorDensity} applied to the distribution $\uniform{2}{4}$, first in very low precision (3 bit exponent, 4 bit significand) and then in half-precision. The theoretical density is plotted alongside a histogram of the relative error incurred when rounding 100,000  samples to low precision (computed in double-precision).  The reported statistic is the K-S (Kolmogorov-Smirnov) test which measures the likelihood that a collection of samples were drawn from a given distribution. This test reports that we cannot reject the hypothesis that the samples are drawn from the corresponding density. Note how in low precision the term in $\frac{1}{(1-t\uro)^2}$ induces a visible asymmetry on the central section of the distribution. This effect is much less pronounced in half-precision.

For low precisions, say up to half-precision, it is computationally feasible to explicitly go through all floating-point numbers and compute the density of the roundoff error distribution $dist$ directly from \cref{eq:LPerrorDensity}. However, this rapidly becomes prohibitively computationally expensive for higher precisions (since the number of floating-point representable numbers grows exponentially).

\subsection{High-Precision Case} \label{subsec:HPerror_dist}

As the working precision increases, a regime changes occurs: on the one hand it becomes practically impossible to enumerate all floating-point representable numbers as done in \cref{eq:LPerrorDensity}, but on the other hand sufficiently well-behaved density functions are numerically close to being constant at the scale of an interval between two floating-point representable numbers. We exploit this smoothness to overcome the combinatorial limit imposed by \cref{eq:LPerrorDensity}.

\begin{theorem}\label{thm:HP_errorDensity}
Let $X$ be a real random variable with PDF $f$. The continuous part $dist_c$ of the distribution of $\erel(X)$ has a PDF given by
$d_c(t) = d_{hp}(t) + R(t)$
where $d_{hp}(t)$ is the function on $[-1,1]$ defined by
\begin{align}
d_{hp}(t)=
\begin{cases}
 \frac{1}{1-t\uro}\sum\limits_{s, e=\emin+1}^{\emax-1} \int_{(-1)^s2^e(1-\uro)}^{(-1)^s2^e(2-u)} \frac{\absv{x}}{2^{e+1}} f(x) ~dx & \absv{t}\leq \frac{1}{2}\\
\\
\frac{1}{1-t\uro}\sum\limits_{s, e=\emin+1}^{\emax-1} \int_{(-1)^s2^e(1-\uro)}^{(-1)^s2^e(\frac{1}{\absv{t}}-\uro)} \frac{\absv{x}}{2^{e+1}} f(x) ~dx  &\frac{1}{2}<\absv{t}\leq 1
\end{cases}\label{eq:HPerrorDensity}
\end{align}
and $R(t)\hspace{-1pt}$ is an error whose total contribution $\absv{R}\hspace{-2pt}\defeq \hspace{-3pt}\int_{-1}^{1}\hspace{-2pt}\absv{R(t)}\hspace{-1pt}dt$ can be bounded by
\begin{align*}
\absv{R}\leq  & ~\Pro{\round(X)=z(s,\emin, k)} +  \Pro{\round(X)=z(s,\emax, k)}  + \\
&\frac{3}{4} \left(\sum_{s, \emin<e<\emax} \absv{f'(\xi_{e,s})\xi_{e,s} + f(\xi_{e,s})}\frac{2^{2e}}{2^{p}}\right)
\end{align*}
where for each exponent $e$ and sign $s$, $\xi_{e,s}$ is a point in $[z(s,e,0), z(s,e,2^{p}-1)]$ if $s=0$ and in $[z(s,e,2^{p}-1), z(s,e,0)]$ if $s=1$.
\end{theorem}

Note how \cref{eq:HPerrorDensity} reduces the sum over \emph{all} floating-point representable numbers in \cref{eq:LPerrorDensity} to a sum over \emph{the exponents} by exploiting the regularity of $f$. This can often be reduced further since one only needs to consider the exponent on the support of $f$. Note also that since $f$ is a PDF, it usually decreases very quickly away from 0, and its derivative decreases even quicker (or vanishes altogether for uniform distributions) and $\absv{R}$ thus tends to be very small.  This is the case for all benchmarks in \cref{sec:evaluation}.  As an example, in single-precision for $X\sim\normal{0}{1}$ we get $\absv{R}<\mathtt{3.2e-7}$, for $X\sim\uniform{-2}{2}$ we get  $\absv{R}<\mathtt{1.2e-7}$; in both cases very close to the smallest floating-point representable increment to 1. Moreover, $\absv{R}\to 0$ as the precision $p\to\infty$.

\cref{eq:HPerrorDensity} is easy to implement, and we present some results in \cref{fig:error_dist} where we have chosen as input: (i) a distribution $\uniform{7}{8}$ where large significands are more likely, (ii) a distribution $\uniform{4}{5}$ where small significands are more likely, (iii) a distribution $\uniform{4}{32}$ where all significands are equally likely, and (iv) a distribution $\normal{0}{1}$ with infinite support. The graphs show the density function given by \cref{eq:HPerrorDensity} in single-precision versus a histogram of the relative error incurred when rounding 1,000,000  samples to single-precision (computed in double-precision).  The K-S test reports that we cannot reject the hypothesis that the samples are drawn from the corresponding distributions.

\subsection{Typical Distribution}\label{subsec:typical}
The distributions depicted in graphs (ii), (v) and (vi) of \cref{fig:error_dist} are extremely similar, despite being computed from very different input distributions. What they have in common is that their input distributions have the property that that all significands in their supports are equally likely.  In fact, we show that under this 
\begin{wrapfigure}{r}{0.48\textwidth}
\vspace{-7mm}
\includegraphics[width=0.52\textwidth]{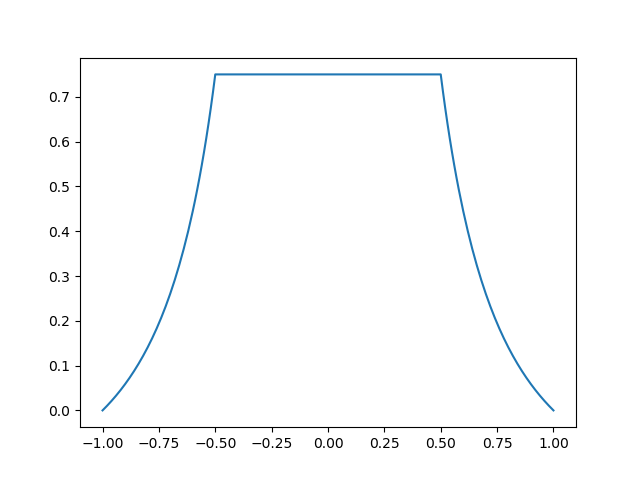}
\captionof{figure}{Typical distribution.}
\label{fig:typical}
\end{wrapfigure}
assumption,  the distribution of roundoff errors given by \cref{eq:LPerrorDensity} converges to a unique density 
as the precision increases, irrespective of the input distribution! 
Since all significands are in practice often equiprobable (it is the case for a third of the benchmarks presented in \cref{sec:evaluation}), this density is of great practical importance.  If one had to choose `the' canonical distribution for roundoff errors, we claim that the density given below should be this distribution, and we therefore call it the \emph{typical distribution}; we depict it in \cref{fig:typical} and formalize it with the following theorem, which makes precise ideas from \cite{dahlqvist2019probabilistic}.

\begin{theorem}\label{thm:typical_dist}
If $X$ is a random variable such that $\Pro{\round(X) = z(s,e,k_0)}=\frac{1}{2^p}$ for any significand $k_0$, then
\begin{align}
d_{typ}(t)\defeq \lim_{p\to\infty} d(t) = \begin{cases}
\frac{3}{4}&\absv{t}\leq\frac{1}{2}
\\
\frac{1}{2}\left(\frac{1}{t}-1\right)+\frac{1}{4}\left(\frac{1}{t}-1\right)^2 & \absv{t}>\frac{1}{2}
\end{cases}\label{eq:typicalpdf}
\end{align} 
where $d(t)$ is the exact density given by \cref{eq:LPerrorDensity}.
\end{theorem}

\subsection{Covariance Structure}\label{subsec:covar}

The result above can be interpreted as saying that if $X$ is such that all mantissas are equiprobable, then $X$ and $\erel(X)$ are asymptotically independent (as $p\to\infty$).  Much more generally,  we now show that if a random variable $X$ has a sufficiently regular PDF, it is close to being uncorrelated from $\erel(X)$. Formally, we prove that the covariance 
\begin{align}
\cov(X,\erel(X)) = \Exp{X.\erel(X)}-\Exp{X}\Exp{\erel(X)}
\label{eq:covdef}
\end{align}
is small, specifically of the order of $\uro$. Note that the expectation in the first summand above is taken \wrt the joint distribution of $X$ and $\erel(X)$.

The main technical obstacles to proving that the expression above is small are that $\Exp{\erel(X)}$ turns out to be difficult to compute (we only manage to bound it) and that the joint distribution
$\Pro{X\in A \wedge \erel(X) \in B}$
does not have a PDF since it is not continuous \wrt the Lebesgue measure on $\R^2$. Indeed, it is supported by the graph of the function $\erel$ which has a Lebesgue measure of 0. This does not mean that it is impossible to compute the expectation
\begin{align}
\Exp{X.\erel(X)} = \int_{\R^2} x\uro t ~d\mathbb{P}\label{eq:expdef}
\end{align}
but it is necessary to use some more advanced probability theory.  We will make the simplifying assumption that the density of $X$ is constant at the level of each interval $\fintvl$ in order to keep the proof manageable.  In practice this is an extremely good approximation. Without this assumption, we would need to add an error term similar to that of \cref{thm:HP_errorDensity} to the expression below. This is not conceptually difficult, but is is rather messy, and it would distract from the main aim of the following theorem which is to bound $\Exp{\erel(X)}$, compute $\Exp{X.\erel(X)}$, and show that the covariance between $X$ and $\erel(X)$ is typically of the order of $\uro$.
\begin{theorem}\label{thm:covar}
If the density of $X$ is piecewise constant on intervals $\fintvl$,  then
\[
\left(L -  \Exp{X}K\frac{\uro}{6}\right) \leq \cov(X,\erel(X))\leq \left(L -  \Exp{X}K\frac{4\uro}{3}\right)
\]
where $L=\sum\limits_{s,e} f((-1)^s2^e)(-1)^s2^{2e}\frac{3\uro^2}{2}$ and $K=\hspace{-12pt} \sum\limits_{s,e=\emin+1}^{\emax-1}\int_{(-1)^s2^e(1-\uro)}^{(-1)^s2^e(2-u)} \frac{\absv{x}}{2^{e+1}} f(x) ~dx$
\end{theorem} 
If the distribution of $X$ is centered (i.e., $\Exp{X}=0$) then $L$ is the exact value of the covariance, and it is worth noting that $L$ is fundamentally an artefact of the floating-point representation and is due to the fact that the intervals $\fintvl[2^e]$ are not symmetric. More generally, for $\Exp{X}$ of the order of, say, 2, the covariance will be small (of the order of $\uro$) as $K\leq 1$ (since $\absv{x}\leq 2^{e+1}$ in each summand).  For very large values of $\Exp{X}$ it is worth noting that there is a high chance that $L$ is also be very large, partially canceling $\Exp{X}$. An illustration of this is given by the \textit{doppler} benchmark examined in \cref{sec:evaluation}, an outlier as it has an input variable with range $\left[20,~20000\right]$. Nevertheless,  even for this benchmark the bounds of \cref{thm:covar} still give a small covariance of the order of $0.001$.

\subsection{Error Terms and P-Boxes}\label{subsec:errorpbox}

In low-precision we can use the exact formula \cref{eq:LPerrorDensity} to compute the error distribution. However, in high-precision,  approximations (typically extremely good) like \cref{eq:HPerrorDensity,eq:typicalpdf} must be used. In order to remain sound in the implementation of our model (see \cref{sec:model}) we must account for the error made by this approximation. We have not got the space to discuss the error made by \cref{eq:typicalpdf}, but taking the term $\absv{R}$ of \cref{thm:HP_errorDensity} as an illustration, we can use the notion of p-box described in \cref{subsec:prob} to create an object which soundly approximates the error distribution. We proceed as follows: since $\absv{R}$ bounds the total error accumulated over all $t\in[-1,1]$, we can soundly bound the CDF $c(t)$ of the error distribution given by \cref{eq:HPerrorDensity} by using the p-box
\[
c^-(t)=\max(0,c(t)-\absv{R})\qquad\text{and} \qquad c^+(t)=\min(1,c(t)+\absv{R})
\]
This p-box is used as the (independent) roundoff error term for the probabilistic model of IEEE arithmetic given by \cref{eq:model}, which we describe in \cref{sec:model}.

\section{Symbolic Affine Arithmetic}
\label{sec:symbolicaffine}

In this section, we introduce \emph{symbolic affine arithmetic}, which we
employ to generate the symbolic form for the roundoff error that we use in
\cref{subsec:conderror}.
Affine arithmetic~\cite{affineoriginal} is a model for range analysis that extends classic interval arithmetic~\cite{intervalarithmetic} with information about linear correlations between operands.
Symbolic affine arithmetic extends standard affine arithmetic by keeping the coefficients of the noise terms \emph{symbolic}.
%
%
We define a \emph{symbolic affine form} as
\begin{align}
\hat{x}= x_{0}+\sum_{i=1}^{n} x_{i}\epsilon_{i},\qquad\text{where}\; \epsilon_{i} \in [-1, 1]. \label{eq:affineform}
\end{align}
We call $x_{0}$ the central symbol of the affine form, while $x_{i}$ are the symbolic coefficients for the noise terms $\epsilon_{i}$.
We can always convert a symbolic affine form to its corresponding interval representation. 
This can be done using interval arithmetic or, to avoid precision loss, using a
global optimizer.

Affine operations between symbolic forms follow the usual rules, such as
\begin{align*} 
\alpha\hat{x}+\beta\hat{y}+\zeta=\alpha{x_{0}}+\beta{y_{0}}+\zeta+\sum_{i=1}^{n}{(\alpha x_{i}+\beta y_{i})\epsilon_{i}} 
\end{align*}
Non-linear operations cannot be represented exactly using an affine form.
Hence, we approximate them like in standard affine arithmetic~\cite{affinebook}.

\noindent \textbf{Sound Error Analysis with Symbolic Affine Arithmetic.}
We now show how the roundoff errors get propagated through the four basic arithmetic operation.
We apply these propagation rules to an arithmetic expression
to accurately keep track of the roundoff errors. 
Since the (absolute) roundoff error directly depends on the range of a computation,  we describe range and error together as a pair
\todo[inline]{Fred: terminology question: this last statement is only true of the \emph{absolute} error.  Is this what \emph{roundoff error} means? If yes, is there another conventional name for the \emph{realtive} error.}
\texttt{(range: Symbol, $\widehat{err}$: Symbolic Affine Form)}.
Here, \texttt{range} represents the infinite-precision range of the computation, while $\widehat{err}$ is the symbolic affine form for the roundoff error in floating-point precision.
Unary operators (e.g., rounding) take as input a (range, error form) pair, and return a new output pair;
binary operators take as input two pairs, one per operand.
For linear operators, the ranges and errors get propagated using the standard rules of affine arithmetic.

For the multiplication, we distribute each term in the first operand to every term in the second operand:
\begin{align}
\nonumber
(\texttt{x},\:\widehat{err}_x)*(\texttt{y}, \:\widehat{err}_y)=(\texttt{x*y},\:\texttt{x}*\widehat{err}_y + \texttt{y}*\widehat{err}_x + \widehat{err}_x*\widehat{err}_y)
\end{align}
The output range is the product of the input ranges and the remaining terms contribute to the error.
%
Only the last (quadratic) expression cannot be represented exactly in symbolic affine arithmetic;
we bound such non-linearities using a global optimizer.
The division is computed as the term-wise multiplication of the numerator with the inverse of the denominator.
Hence, we need the inverse of the denominator error form, and then we can proceed as for multiplication.
To compute the inverse, we leverage the symbolic expansion used in FPTaylor~\cite{solovyev2018rigorous}.

Finally, after every operation we apply the unary rounding operator from \cref{eq:traditional}. The infinite-precision range is not affected by rounding. The rounding operator appends a fresh noise term to the symbolic error form. 
The coefficient for the new noise term is the (symbolic) floating-point range given by the sum of the input range with the input error form.
\section{Algorithm and Implementation}\label{sec:model}

\begin{figure}[tb]
	\centering
	\includegraphics[width=\textwidth]{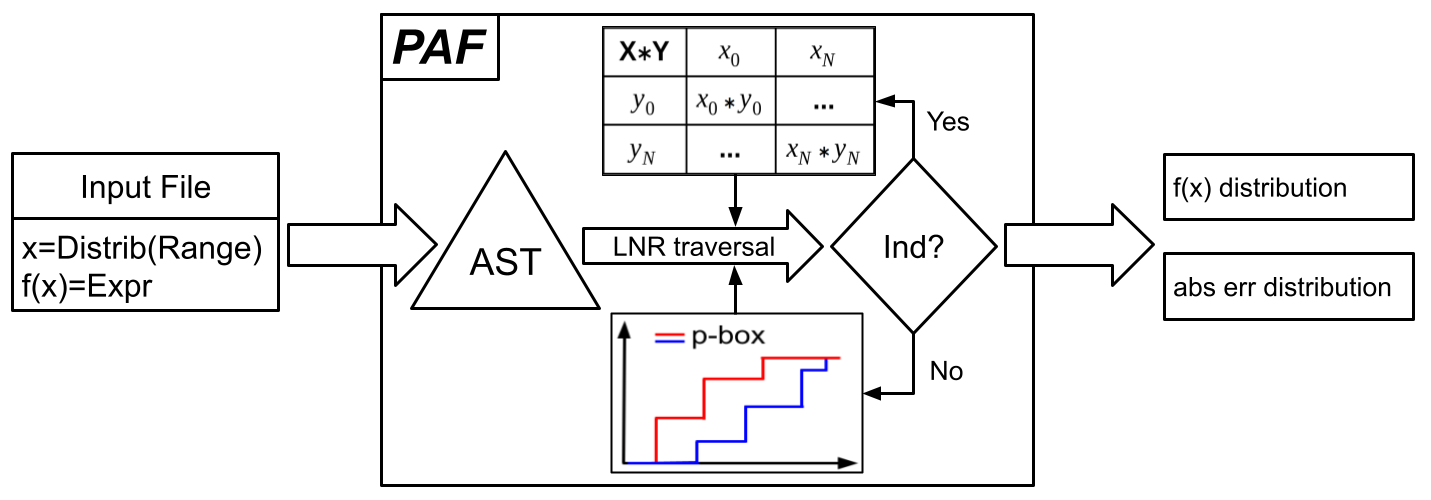}
	\caption{Toolflow of \Tool.}
	\label{fig:toolflow}
\end{figure}

In this section, we describe our probabilistic model of floating-point
arithmetic and how we implement it in a prototype named \Tool (for
Probabilistic Analysis of Floating-point errors).
Fig.~\ref{fig:toolflow} shows the toolflow of \Tool.
%

%
%
%


\subsection{Probabilistic model}\label{subsec:model}
\Tool takes as input a text file describing a probabilistic floating-point
computation and its input distributions.
The kinds of computations we support are captured with this simple grammar:
\[
\tt t::= z \mid x_i \mid t\mop t \qquad \mathtt{z}\in\F, \mathtt{i}\in\N, \mop\in\{+,-,\times,\div\}
\]
Following ~\cite{K81c}, we interpret each computation $\tt  t$ given by the grammar as a random variable.
We define the interpretation map $\sem{-}$ over the computation tree inductively.
The base case is given by
$\sem{\mathtt{z}(s,e,k)} \defeq (-1)^s2^e(1+k2^{-p})$ and $\sem{\mathtt{x_i}} \defeq X_i$,
where the real numbers $\sem{\mathtt{z}(s,e,k)}$ are understood as constant random variables and each $X_i$ is a random input variable
with a user-specified distribution.
Currently, \Tool supports several well-known distributions out-of-the-box (e.g., uniform, normal, exponential, beta), and a user can also define custom distributions as piecewise functions.
For the inductive case $\sem{\tt t_1\mop t_2}$, we put the lessons from \cref{sec:errordist} to work.
Recall first the probabilistic model \cref{eq:probabilistic}:
\[
x\mop y=(x\iop y)(1+\delta), \qquad\delta\sim dist
\]
In \cref{subsec:LPerror_dist}, we showed that $dist$ should be taken as the distribution of the actual roundoff errors of the random elements $(x\iop y)$.  We therefore define:
\begin{align}
\sem{\mathtt{t_1 \mop t_2}}\defeq(\sem{\mathtt{t_1}}\iop\sem{\mathtt{t_2}})\times (1+ \erel(\mathtt{\sem{t_1}\iop\sem{\tt t_2})}).\label{eq:model}
\end{align}
To evaluate the model of \cref{eq:model}, we first use the appropriate closed-form expression \cref{eq:LPerrorDensity,eq:HPerrorDensity,eq:typicalpdf}  derived in \cref{sec:errordist} to evaluate the (distribution of the) random variable $\erel(\mathtt{\sem{t_1}\iop\sem{\tt t_2})}$ --- or the corresponding p-box as described in \cref{subsec:errorpbox}. We then use \cref{thm:covar} to justify evaluating the multiplication operation in \cref{eq:model} \emph{independently} --- that is to say by using \cref{eq:pdfarithm} --- since the roundoff process is very close to being uncorrelated to the process generating it. The validity of this assumption is also confirmed experimentally by the remarkable agreement of Monte-Carlo simulations with this analytical model (see \cref{sec:evaluation}).

%
We now introduce the algorithm for evaluating the model given in \cref{eq:model}.
The evaluation performs an in-order (LNR) traversal of the \emph{Abstract
Syntax Tree} (AST) of a computation given by our grammar, and it feeds the
results to the parent level along the way.
At each node, it computes the probabilistic range of the intermediate result
using the probabilistic ranges computed for its children nodes (i.e.,
operands).
We first determine whether the operands are independent or not (Ind?\ branch in the toolflow),
and we either apply a cheaper (i.e., no SMT solver invocations) algorithm if they are independent (see below) or
a more involved one (see \cref{subsec:dependentOperation}) if they are not.
We describe our methodology at a generic intermediate computation in the AST of the expression.

We consider two distributions $X$ and $Y$ discretized into DS-structures $DS_{X}$ and $DS_{Y}$ (\cref{subsec:prob}), and we want to derive the DS-structure $DS_Z$ for $Z=X\iop Y$,  $\iop \in \{+,-,\times,\div\}$.\todo[inline]{Fred: should we say how it's done at the leaves?}
Together with the DS-structures of the operands, we also need the traces $trace_X$ and $trace_Y$ containing the history of the operations performed so far, one for each operand. A trace is constructed at each leaf of the AST with the input distributions and their range. It is then propagated to the parent level and populated at each node with the current operation. 
Such history traces are critical when dealing with dependent operations since they allow us to interrogate an SMT solver about the feasibility of the current operation, as we describe in the next section.
When the operands are independent, we simply use the arithmetic operations on independent DS-structures~\cite{bouissou2012generalization}
(see \cref{sec:appdx-independent}).

\subsection{Computing Probabilistic Ranges for Dependent Operands}\label{subsec:dependentOperation}
When the operands are dependent, we start by assuming that the dependency is unknown.
%
This assumption is sound because the dependency of the operation is included in the set of unknown dependencies, while the result of the operation is no longer a single distribution but a p-box.
Due to this ``unknown assumption'', the CDFs of the output p-box are a very pessimistic over-approximation of the operation, i.e., they are far from each other.
%
%
Our key insight is to then leverage an SMT solver to prune infeasible combinations of intervals from the input DS-structures, which prunes regions of zero probability from the output p-box.
%
%
%
This probabilistic pruning using a solver squeezes together the CDFs of the output p-box, often resulting in a much more accurate over-approximation.
Thanks to using a solver, we move from an unknown to a \emph{partially known} dependency between the operands.
Currently, \Tool supports the Z3~\cite{Z3} and dReal~\cite{dReal} SMT solvers.

\begin{algorithm}[tb]
	\caption{Dependent Operation $Z=X\iop Y$}\label{dependentalg}
	\begin{algorithmic}[1]
		\Function {dep\_op}{$DS_X,\iop,DS_Y,trace_X,trace_Y$}
		\State $DS_Z=list()$
		\ForAll {$ ([x_1, x_2], p_{x}) \in DS_X$} 
		\ForAll {$ ([y_1, y_2], p_{y}) \in DS_Y$}
		\State $[z_1, z_2]=[x_1, x_2] \iop\;[y_1, y_2]$ \Comment{operation between intervals}
		\State $[z_1', z_2']=SMT.prune([z_1, z_2])$
		\If {$SMT.check(trace_X \land trace_Y \land [x_1, x_2] \land [y_1, y_2])$ \textbf{is} $SAT$}
		\State $p_Z=\mbox{unknown-probability}$
		\Else 
		\State $p_Z=0$
		\EndIf
		\State $DS_Z.append(([z_1', z_2'],p_Z))$
		\EndFor
		\EndFor
		\State $trace_Z=trace_X \cup trace_Y\cup\{Z=X\iop Y\}$
		\State \textbf{return} $DS_Z, trace_Z$
		\EndFunction
	\end{algorithmic}
\end{algorithm}

\cref{dependentalg} shows the pseudocode of our algorithm for computing the
probabilistic output range (i.e., DS-structure) for dependent operands.
When dealing with dependent operands, interval arithmetic (line 5) might not be as precise as in the independent case.
Hence, we use an SMT solver to prune away any over-approximations introduced by
interval arithmetic when computing with dependent ranges (line 6);
this use of the solver is orthogonal to the one dealing with probabilities.
On line 7, we check with an SMT solver whether the current combination of
ranges $[x_1, x_2]$ and $[y_1, y_2]$ is compatible with the traces of the
operands.
If the query is satisfiable, the probability is strictly greater than zero but currently unknown (line 8).
If the query is unsatisfiable, we assign a probability of zero to the range in $DS_Z$ (line 10). 
%
%
%
%
Finally, we append a new range to the DS-structure $DS_Z$ (line 11).
Note that the loops are independent, and hence in our prototype implementation we run them in parallel.

After this algorithm terminates, we still need to assign probability values to
all the unknown-probability ranges in $DS_Z$.
Since we cannot assign an exact value, we compute a range of potential values
$[p_{z_{min}}, p_{z_{max}}]$ instead. 
This computation can be encoded as a \emph{linear programming} routine exactly
as described in previous work~\cite{bouissou2012generalization} (see \cref{sec:appdx-lp}).


%

%
%
%

%
%
%
%

\subsection{Computing Conditional Roundoff Error}\label{subsec:conderror}

The final step of our toolflow computes the conditional roundoff error by
combining the symbolic affine arithmetic error form of the computation (see
\cref{sec:symbolicaffine}) with the probabilistic range analysis described
above.
The symbolic error form gets maximized conditioned on the results of all the
intermediate operations landing in the given confidence interval (e.g., 99\%)
of their respective ranges (computed as described in the previous section).
%
%
Note that conditioning only on the last operation of the computation tree (i.e., the AST root) would lead to extremely pessimistic over-approximation since all the outliers in the intermediate operations would be part of the maximization routine.
%
This would lead to our tool \Tool computing pessimistic error bounds typical of
worst-case analyzers.

\begin{algorithm}[tb]
	\caption{Conditional Roundoff Error Computation}\label{conderror}
	\begin{algorithmic}[1]
		
		\Function {cond\_err}{$DSS, errorForm, confidence$}
		\State {$allRanges=list()$}
		\ForAll {$DS_i \in DSS$}
		\State {$focals=sorted(DS_i,\:key=prob,\: order=descending)$}
		\State {$accumulator=0$}
		\State {$ranges=\text{\O}$}
		\ForAll {$ ([x_1, x_2], p_{x}) \in focals$}
			\State {$accumulator=accumulator+p_{x}$}
			\State {$ranges=ranges \cup [x_1, x_2]$}
			\If {$accumulator \ge confidence$}
				\State $allRanges.append(ranges)$
				\State $\mathbf{break}$
			\EndIf
		\EndFor
		\EndFor
		\State $error=maximize(errorForm, allRanges)$
		\State \textbf{return} $error$
		\EndFunction
	\end{algorithmic}
\end{algorithm}

\cref{conderror} shows the pseudocode of the roundoff error computation
algorithm.
The algorithm takes as input a list $DSS$ of DS-structures (one for each
intermediate result range in the computation), the generated symbolic error
form, and a confidence interval.
It iterates over all the intermediate DS-structures (line 3), and for each it
determines the ranges needed to support the chosen confidence intervals (lines
4--12).
In each iteration, it first sorts the list of range-probability pairs (i.e.,
focal elements) of the current DS-structure by their probability value in a
descending order (line 4).
This is a heuristic that prioritizes the focal elements with most of the
probability mass (aka \emph{typicals}) to avoid including the unlikely outliers
(with low probability) that cause large roundoff errors into the final roundoff
error computation.
With the help of an accumulator (line 8), we keep collecting focal elements
(line 9) until the accumulated probability satisfies the confidence interval
(line 10).
Finally, we maximize the error form conditioned to the collected ranges of all
intermediate operations (line 13).
In \Tool, the maximization is done using the rigorous global optimizer
Gelpia~\cite{gelpia}.

\section{Experimental Evaluation}
\label{sec:evaluation}

We evaluate \Tool on the standard FPBench
benchmark suite~\cite{fpbench,fpbench-web} that uses the four basic operations we
currently support $\{+,-,\times,\div\}$.
Many of these benchmarks were also used in recent related work~\cite{probdaisy}
that we compare against.
The benchmarks come from a variety of domains: embedded software (\emph{bsplines}), linear classifications (\emph{classids}), physics computations (\emph{dopplers}), filters (\emph{filters}), controllers (\emph{traincars}, \emph{rigidBody}), polynomial approximations of functions (\emph{sine}, \emph{sqrt}), solving equations (\emph{solvecubic}), and global optimizations (\emph{trids}).
Since FPBench has been primarily used for worst-case roundoff error
analysis,  the benchmarks come with ranges for input variables, but they do not specify
input distributions.
We instantiate the benchmarks with three well-known distributions for all the
inputs: uniform, standard normal distribution,  and double exponential (\ie Laplace)
distribution with $\sigma=0.01$ which we will call `exp'. 
The normal and exp distributions get truncated to the given range.
We assume single-precision floating-point format for all
operands and operations, but it is straightforward to
extend \Tool to mixed-precision computations.

To assess the accuracy and performance of \Tool, we compare it with
PrAn~\cite{probdaisy}, the current state-of-the-art tool for automated analysis
of probabilistic roundoff errors.
PrAn currently supports only uniform and normal input distributions.
It offers six different tool configurations; for each benchmark, we run
all of them and report the best result.
We fix the number of intervals in each discretization to 50 to match PrAn.
We choose 99\% as the confidence interval for the computation of our conditional roundoff error (\cref{subsec:conderror}) and of PrAn's probabilistic error.
We also compare our probabilistic error bounds against FPTaylor which performs worst-case roundoff error analysis, and hence it does not
take into account the distributions of the input variables.
We ran our experiments in parallel on a 4-socket 2.2 GHz 8-core Intel Xeon
E5-4620 machine with 128 GB of memory.

\newcolumntype{L}[1]{>{\raggedright\let\newline\\\arraybackslash\hspace{0pt}}m{#1}}

\begin{table*}[t!]
	\centering
	\scriptsize
	\caption{Roundoff error bounds reported by \Tool, PrAn, and FPTaylor given
	uniform (uni), normal (norm), and Laplace (exp) input distributions.
	We set the confidence interval to 99\% for \Tool and PrAn, and mark the smallest reported
	roundoff errors for each benchmark in bold. Asterisk (*) highlights a difference of more than one order of magnitude between \Tool and FPTaylor.}
\label{erroranalysis}
	\renewcommand{\arraystretch}{1.1}
	\begin{tabular}{@{\extracolsep{2.3pt}}p{1.5cm}L{1.2cm}L{1.2cm}L{1.3cm}L{1.3cm}L{1.3cm}L{1.3cm}L{1.3cm}@{}}
		\toprule
		\multirow{4}{*}{Benchmark} & \multicolumn{2}{c}{uniform} & \multicolumn{2}{c}{normal} & \multicolumn{1}{c}{exp} &\multirow{4}{1.2cm}{FpTaylor}\\
		\cmidrule{2-3} \cmidrule{4-5} \cmidrule{6-6}
		& \multicolumn{1}{c}{\Tool} & \multicolumn{1}{c}{PrAn} & \multicolumn{1}{c}{\Tool} & \multicolumn{1}{c}{PrAn} & \multicolumn{1}{c}{\Tool} \\
		\midrule
		bspline0 & \textbf{5.71e-08} & 6.12e-08 & \textbf{5.71e-08} & 6.12e-08 & \textbf{5.71e-08} &5.72e-08 \\
		\mydashline{}
		bspline1 & \textbf{1.86e-07} & 2.08e-07 & \textbf{1.86e-07} & 2.08e-07 & \textbf{6.95e-08} &1.93e-07 \\
		\mydashline{}
		bspline2 & \textbf{1.94e-07} & 2.13e-07 & \textbf{1.94e-07} & 2.13e-07 & \textbf{2.11e-08} &2.10e-07 \\
		\mydashline{}
		bspline3 & \textbf{4.22e-08} & 4.65e-08 & \textbf{4.22e-08} & 4.65e-08 & \textbf{7.62e-12}* &4.22e-08 \\
		\mydashline{}
		classids0 & \textbf{6.93e-06} & 8.65e-06 & \textbf{4.45e-06} & 8.64e-06 & \textbf{1.70e-06} &6.85e-06 \\
		\mydashline{}
		classids1 & \textbf{3.71e-06} & 4.63e-06 & \textbf{2.68e-06} & 4.62e-06 & \textbf{7.62e-07} &3.62e-06 \\
		\mydashline{}
		classids2 & \textbf{5.23e-06} & 7.32e-06 & \textbf{3.85e-06} & 7.32e-06 & \textbf{1.46e-06} &5.15e-06 \\
		\mydashline{}
		doppler1 & \textbf{7.95e-05} & 1.17e-04 & \textbf{5.08e-07}* & 1.17e-04 & \textbf{4.87e-07}* &6.10e-05 \\
		\mydashline{}
		doppler2 & \textbf{1.43e-04} & 2.45e-04 & \textbf{6.61e-07}* & 2.45e-04 & \textbf{6.28e-07}* &1.11e-04 \\
		\mydashline{}
		doppler3 & \textbf{4.55e-05} & 5.12e-05 & \textbf{9.11e-07}* & 5.12e-05 & \textbf{8.95e-07}* &3.41e-05 \\
		\mydashline{}
		filter1 & \textbf{1.25e-07} & 2.03e-07 & \textbf{1.25e-07} & 2.03e-07 & \textbf{5.43e-09}* &1.25e-07 \\
		\mydashline{}
		filter2 & \textbf{7.93e-07} & 1.01e-06 & \textbf{6.13e-07} & 1.01e-06 & \textbf{2.90e-08}* &7.93e-07 \\
		\mydashline{}
		filter3 & \textbf{2.34e-06} & 2.86e-06 & \textbf{2.05e-06} & 2.87e-06 & \textbf{1.09e-07}* &2.23e-06 \\
		\mydashline{}
		filter4 & \textbf{4.15e-06} & 5.20e-06 & \textbf{4.15e-06} & 5.20e-06 & \textbf{4.61e-07} &3.81e-06 \\
		\mydashline{}
		rigidbody1 & 1.74e-04 & \textbf{1.58e-04} & \textbf{6.14e-06}* & 1.58e-04 & \textbf{4.80e-07}* &1.58e-04 \\
		\mydashline{}
		rigidbody2 & 1.96e-02 & \textbf{9.70e-03} & \textbf{5.99e-05}* & 9.70e-03 & \textbf{9.55e-07}* & 1.94e-02 \\
		\mydashline{}
		sine & \textbf{2.37e-07} & 2.40e-07 & \textbf{2.37e-07} & 2.40e-07 & \textbf{1.49e-08}* &2.38e-07 \\
		\mydashline{}
		solvecubic & \textbf{1.78e-05} & 1.83e-05 & \textbf{6.84e-06} & 1.83e-05 & \textbf{2.76e-06} & 1.60e-05 \\
		\mydashline{}
		sqrt & \textbf{1.54e-04} & \textbf{1.54e-04} & \textbf{1.10e-06}* & 1.54e-04 & \textbf{2.46e-07}* &1.51e-04 \\
		\mydashline{}
		traincars1 & \textbf{1.76e-03} & 1.96e-03 & \textbf{8.26e-04} & 1.96e-03 & \textbf{4.50e-04} &1.74e-03 \\
		\mydashline{}
		traincars2 & \textbf{1.04e-03} & 1.36e-03 & \textbf{3.61e-04} & 1.36e-03 & \textbf{2.83e-05}* &9.46e-04 \\
		\mydashline{}
		traincars3 & \textbf{1.75e-02} & 2.29e-02 & \textbf{9.56e-03} & 2.29e-02 & \textbf{8.95e-04}* &1.80e-02 \\
		\mydashline{}
		traincars4 & \textbf{1.81e-01} & 2.30e-01 & \textbf{8.87e-02} & 2.30e-01 & \textbf{7.33e-03}* &1.81e-01 \\
		\mydashline{}
		trid1 & \textbf{6.01e-03} & 6.03e-03 & \textbf{1.58e-05}* & 6.03e-03 & \textbf{1.58e-05}* &6.06e-03 \\
		\mydashline{}
		trid2 & \textbf{1.03e-02} & 1.17e-02 & \textbf{2.42e-05}* & 1.17e-02 & \textbf{2.43e-05}* &1.03e-02 \\
		\mydashline{}
		trid3 & \textbf{1.75e-02} & 1.95e-02 & \textbf{6.80e-05}* & 1.95e-02 & \textbf{6.77e-05}* &1.75e-02 \\
		\mydashline{}
		trid4 & \textbf{2.69e-02} & 2.88e-02 & \textbf{2.64e-04}*& 3.03e-02 & \textbf{2.64e-04}* & 2.66e-02 \\
		\bottomrule
	\end{tabular}
\end{table*}

Table~\ref{erroranalysis} compares roundoff errors reported by \Tool, PrAn, and
FPTaylor.
\Tool outperforms PrAn by computing tighter probabilistic error bounds on
almost all benchmarks, occasionally by orders of magnitude.
In the case of uniform input distributions, \Tool provides tighter bounds for
20 out of 27 benchmarks, for 6 benchmarks the bounds from PrAn are tighter,
while for \emph{sqrt} they are the same.
In the case of normal input distributions, \Tool provides tighter bounds for 26
out of 27 benchmarks, while for \emph{bspline3} the bounds from PrAn are
tighter.
%
%
Unlike PrAn, \Tool supports probabilistic output range analysis as well.
Table~\ref{rangeanalysis} in \cref{sec:appdx-evaluation} presents detailed
range analysis results.

In Table~\ref{erroranalysis}, of particular interest are benchmarks (6 for
normal and 10 for exp) where the error bounds generated by \Tool for the 99\%
confidence interval are at least an order of magnitude tighter than the
worst-case bounds generated by FPTaylor.
For such a benchmark and input distribution, \Tool's results inform a user that
there is an opportunity to optimize the benchmark (e.g., by reducing precision
of floating-point operations) if their use-case can handle at most 1\% of
inputs generating roundoff errors that exceed a user-provided bound.
FPTaylor's results, on the other hand, do not allow for a user to explore such
fine-grained trade-offs since they are worst-case and do not take probabilities
into account.

\begin{figure}[tb]
	\centering
	\begin{tabular}{l l}
		\includegraphics[width=0.5\textwidth]{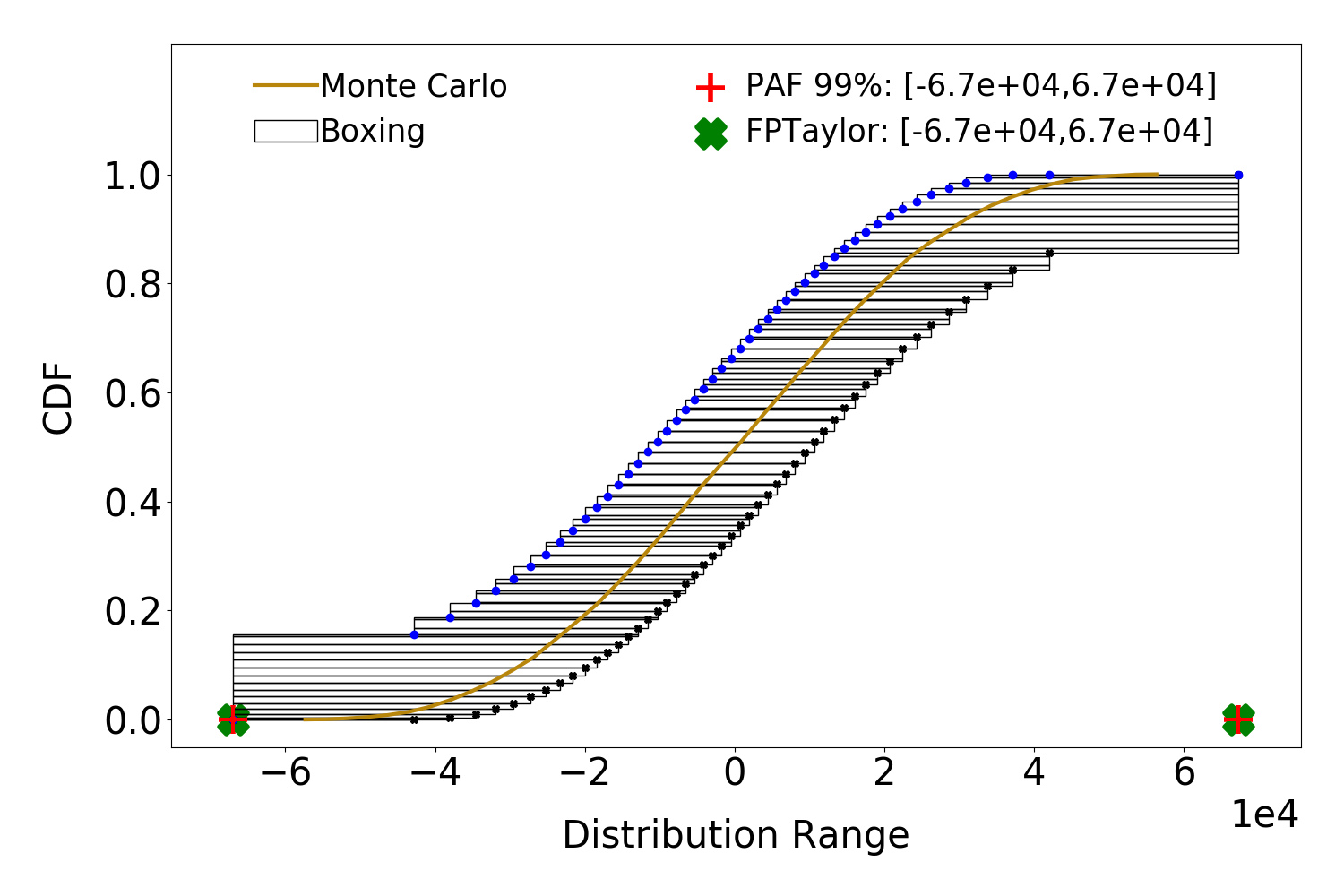}
		&
		\includegraphics[width=0.5\textwidth]{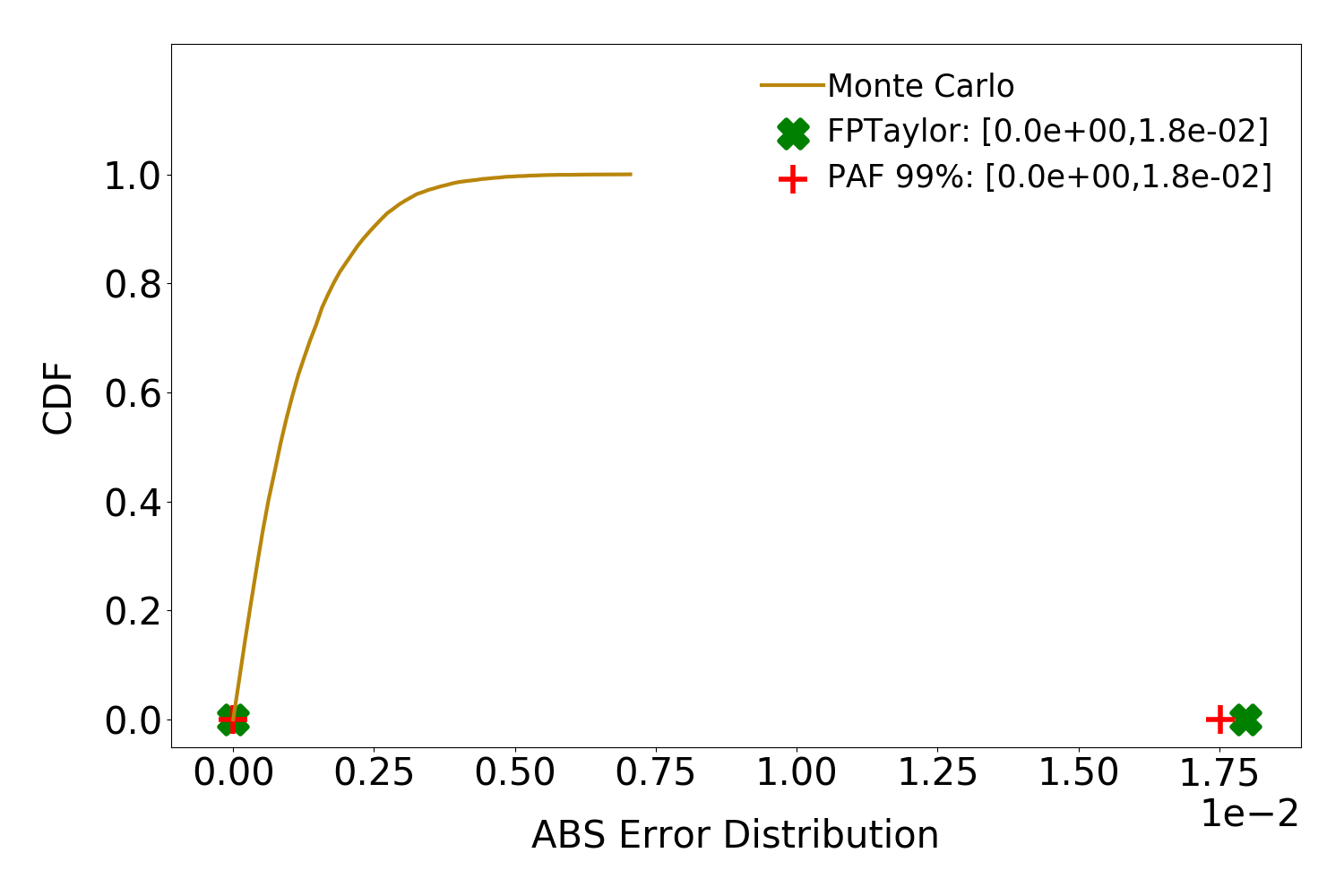}
		\\
		\includegraphics[width=0.5\textwidth]{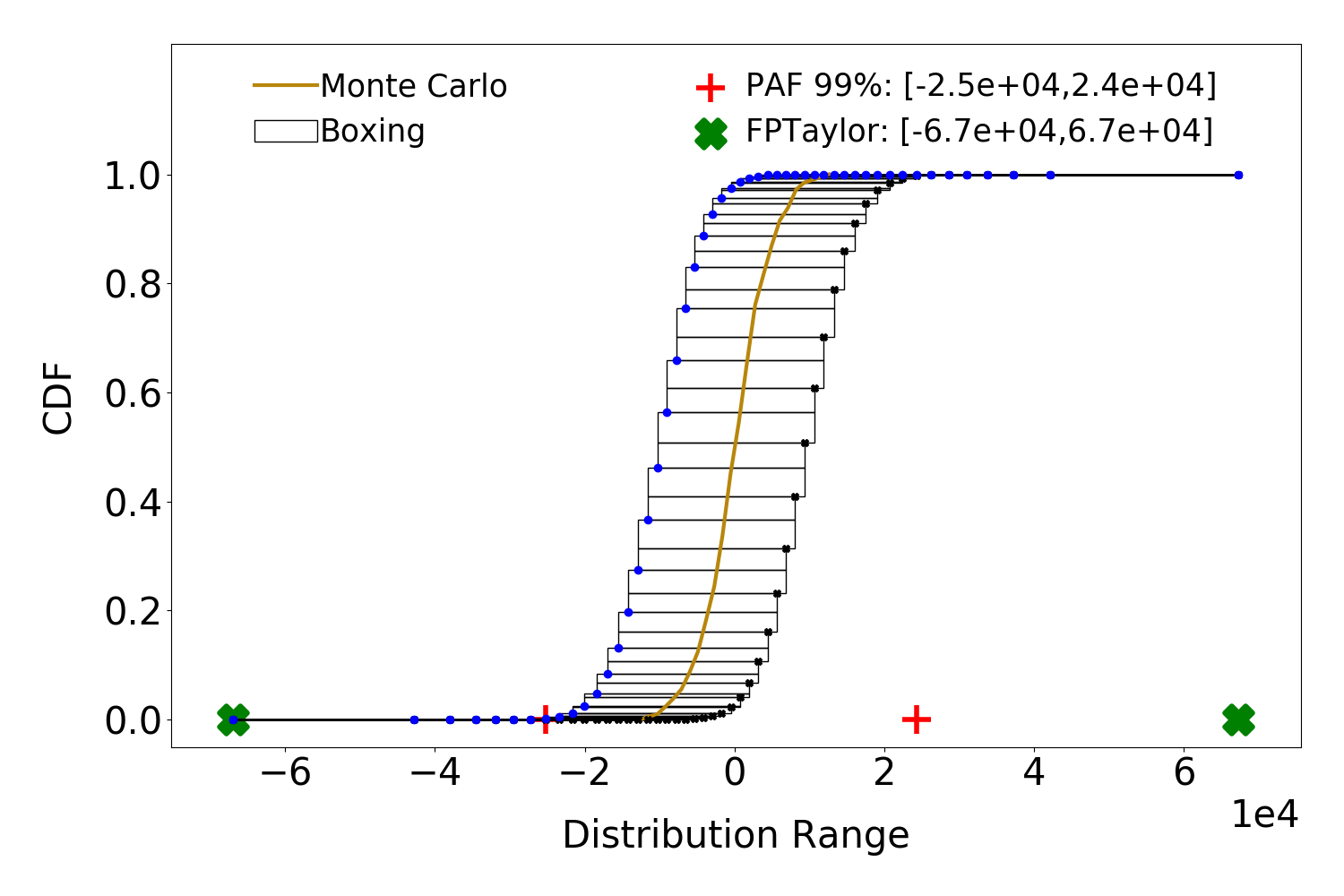}
		&
		\includegraphics[width=0.5\textwidth]{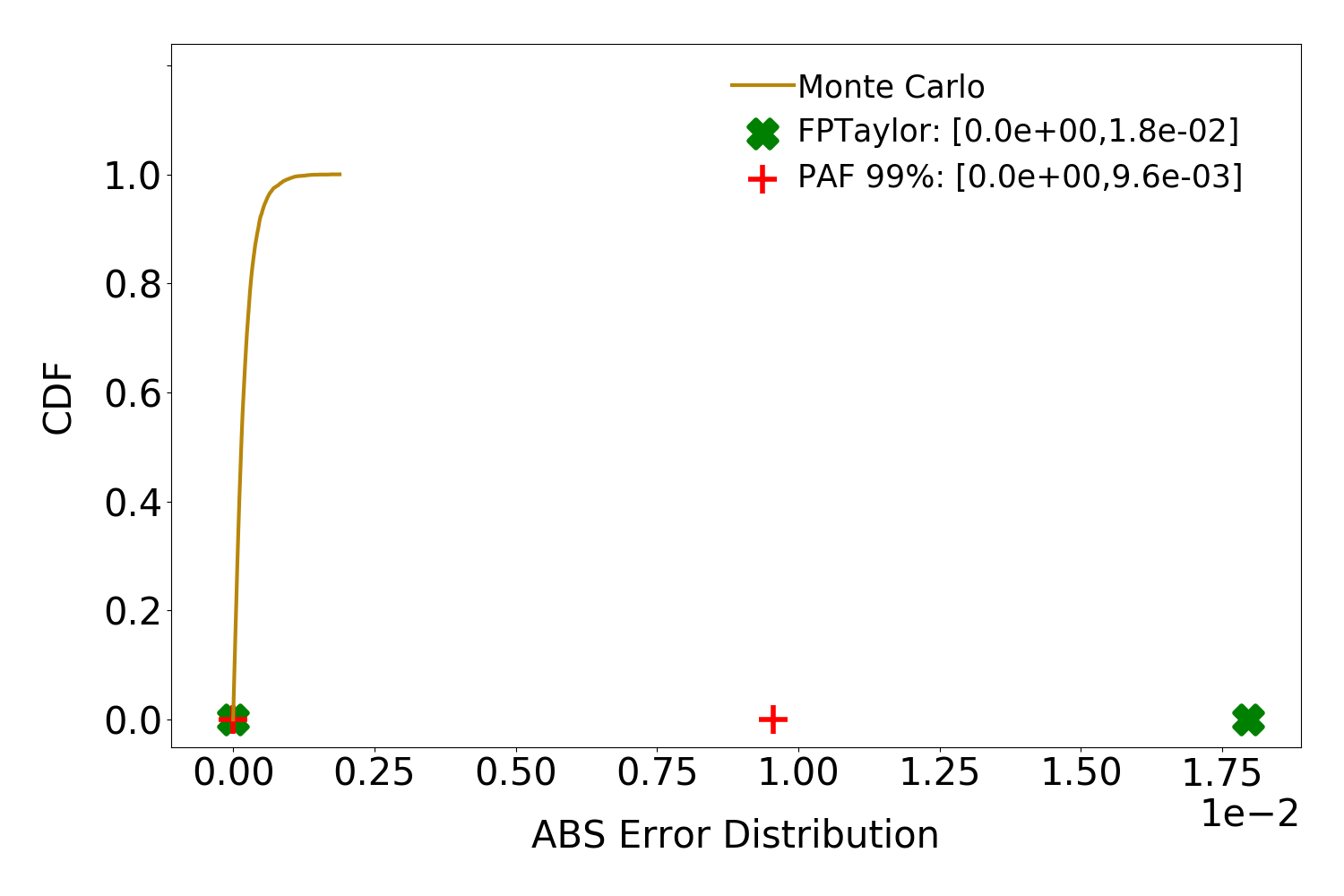}
		\\
		\includegraphics[width=0.5\textwidth]{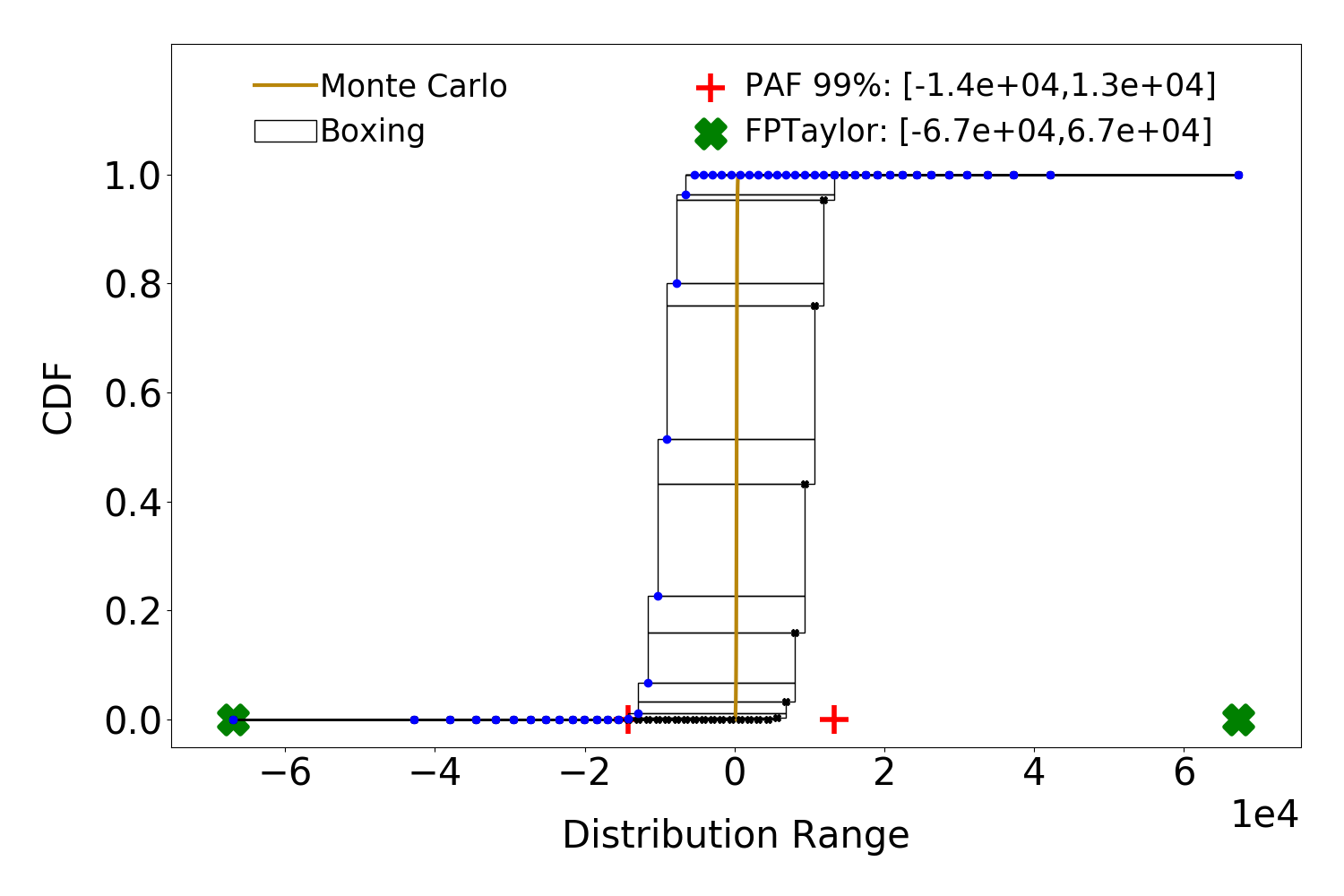}
		&
		\includegraphics[width=0.5\textwidth]{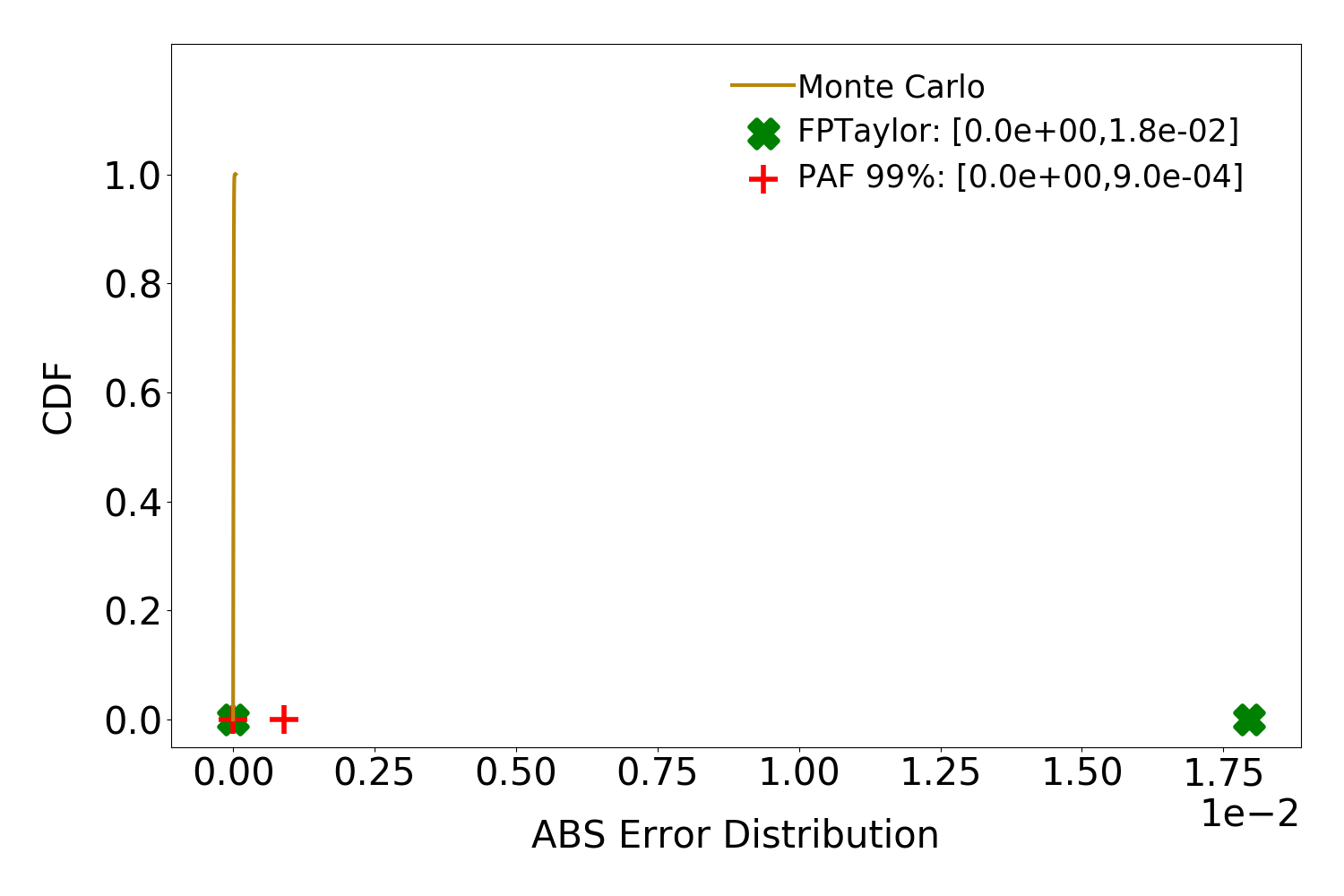}
	\end{tabular}
	\caption{CDFs of the range (left) and error (right) distributions for the benchmark \emph{train3} for uniform (top), normal (center), and exp (bottom).}
	\label{fig:range_error_traincar}
\end{figure}

In general, we see a gradual reduction of the errors transitioning from uniform to normal to exp.
When the input distributions are uniform, there is a significant chance of generating a roundoff error of the same order of magnitude as the worst-case error, because all the inputs, including the extrema, are equally likely.
The standard normal distribution concentrates more than 99\% of probability mass in the interval $[-3, 3]$, thus resulting in the \emph{long tail} phenomenon, where less than 0.5\% of mass spreads in the interval $[3, \infty]$.\todo[inline]{Fred: is the input range is [20,20000] do you take the standard normal centred around the mid point? the leftmost point? around zero? Same question for exp.}
When the normal distribution gets truncated in a neighborhood of zero (e.g., $[0,1]$ for \emph{bsplines} and \emph{filters}) nothing changes with respect to the uniform case --- there is still a high chance of committing errors close to the worst-case.
However, when the normal distribution gets truncated in a wider range (e.g., $[-100, 100]$ for \emph{trids}) and we do have long tails, then the outliers causing large errors are very rare events, not included in the 99\% confidence interval.
The exponential distribution further compresses the 99\% probability mass in the tiny interval $[-0.01, 0.01]$, so the long tails effect is common among all the benchmarks. 
%
%
%
%
%
%
%

%

%
%
%
The runtimes of \Tool vary between 10 minutes for small benchmarks, such as
\emph{bsplines}, to several hours for benchmarks with more than 30 operations,
such as \emph{trid4}; they are always less than two hours, except for
\emph{trids} with 11 hours and \emph{filters} with 6 hours.
The runtime of \Tool is usually dominated by Z3 invocations, and the long
runtimes are caused by numerous Z3 timeouts that the respective benchmarks
induce.
The runtimes of PrAn are comparable to \Tool since they are always less than
two hours, except for \emph{trids} with 3 hours, \emph{sqrt} with 3 hours, and
\emph{sine} with 11 hours.
Note that neither \Tool nor PrAn are memory intensive, and hence memory
consumption is not an issue.

To assess the quality of our rigorous (i.e., sound) results, we implement
Monte Carlo sampling to generate both roundoff error and output range
distributions.
The procedure consists of randomly sampling from the provided input
distributions, evaluating the floating-point computation in both the specified
and high-precision (e.g., double-precision) floating-point regimes to measure
the roundoff error, and finally partitioning the computed errors into bins to
get an approximation (i.e., histogram) of the PDF.
Of course, Monte Carlo sampling does not provide rigorous bounds, but is a
useful tool to assess how far the rigorous bounds computed statically by \Tool
are from an empirical measure of the error.

Fig.\ref{fig:range_error_traincar} shows the effects of the input distributions
on the output and roundoff error ranges of the \emph{traincars3} benchmark.
In the error graphs (right column), we show the Monte Carlo sampling evaluation
(yellow line) together with the error bounds from \Tool with 99\% confidence
interval (red plus symbol) and FPTaylor's worst-case bounds (green crossmark). 
In the range graphs (left column), we also plot \Tool's p-box
over-approximations.
We can observe that in the case of uniform inputs the computed p-boxes overlap
at the extrema of the output range.
This phenomenon makes it impossible to distinguish between 99\% and 100\%
confidence intervals, and hence as expected the bound reported by \Tool is
almost identical to FPTaylor's.
This is not the case for normal and exponential distributions, where we can
observe the long tail phenomenon, and \Tool can significantly improve both the
output range and error bounds over FPTaylor.
%
%
Hence, we again illustrate how pessimistic the bounds from worst-case tools can be when the information about the input distributions
is not taken into account.
%
%
Finally, the graphs illustrate how the rigorous p-boxes and error bounds from
\Tool follow their respective empirical estimations, showing that \Tool adjusts
them based on the shape of the input distribution.

\section{Related Work}
\label{sec:related-work}

Our work draws inspiration from the seminal probabilistic affine arithmetic
approach of Bouissou et al.~\cite{bouissou2012generalization}.
The fundamental object in this approach is the \emph{probabilistic affine
form}, which is an affine form whose error symbols are associated with p-boxes
(see \cref{subsec:prob}).
The authors propose an innovative technique to perform dependent operations
between random variables based on affine arithmetic. 
However, their approach can only detect affine dependencies between the
operands.
%
%
On the other hand, in \Tool, we detect dependencies between operands using an SMT solver. 
Our approach allows us to handle not only polynomial dependencies, but also any non-linearities the solver supports.
%
%
%
%
Finally, their approach only computes output ranges, while we went a step
further and can compute roundoff error bounds as well, which is a non-trivial
extension.

The most similar approach to our work is the recent static probabilistic
roundoff error analysis called PrAn~\cite{probdaisy}; to the best of our
knowledge, this is also the only work apart from this paper that presents a
rigorous and general probabilistic roundoff error analysis.
PrAn builds on the probabilistic affine arithmetic of Bouissou et
al.~\cite{bouissou2012generalization}, and inherits the same limitations in
dealing with dependent operations we described above. 
%
%
Like us, their method hinges on a discretization scheme that builds
p-boxes for both the input and error distributions and propagates them through
the computation.
The question of how these p-boxes are chosen
is left open.  In contrast,  we take the input variables to be
random variables that the user specifies and show how the distribution of
each error terms can be computed directly and exactly from the random variable
generating it (see \cref{sec:errordist}).
%
%
%
%
Furthermore, unlike PrAn, \Tool leverages the noncorrelation between random variables and the corresponding error distribution we described in \cref{subsec:covar}.
The immediate consequence is that \Tool performs the rounding in \cref{eq:probabilistic} as an \emph{independent} operation, thus without any uncertainty.
%
%
Putting these together leads to \Tool computing tighter probabilistic roundoff
error bounds than PrAn, as our experimental results show (see
\cref{sec:evaluation}).

The idea of using a probabilistic model of rounding errors to analyze \emph{deterministic} computations can be traced back to \cite{von1947numerical}.  Parker's so-called `Monte Carlo arithmetic' \cite{parker2000monte} is probably the most detailed description of this approach.  We, however,  consider \emph{probabilistic} computations,  and are therefore able to derive the error distribution from first principles.
For this reason, the famous critique of the probabilistic approach to roundoff errors \cite{kahan1996improbability} does not apply to this work, since it only considers deterministic computations.  When dealing with probabilistic computations however, everything -- including rounding errors -- has a probabilistic behavior and the probabilistic approach is therefore a  necessity.  

More recently, probabilistic roundoff error models have been investigated by Higham and Mary~\cite{higham2019new} as well as Ipsen and Zhou~\cite{ipsen2019probabilistic}. Interestingly, because these authors are interested in large-dimensional problems, neither approach needs to explicitly specify the distribution of errors.  However, this means that these approaches are only applicable to large-dimensional problems, and are completely unsuited for the domains (e.g., control systems, filters, cyber-physical systems)  captured by the FPBench benchmark suite. Finally, unlike us, neither approach is sound since they rely on concentration of measure inequalities.

%

Unlike probabilistic roundoff error analysis, worst-case analysis of roundoff errors has been an active research area
with numerous published approaches\\
\cite{gappa,gappapp,smartfloat,fluctuat,rangelab,precisa,darulova2018daisy,2015_fm_sjrg,solovyev2018rigorous,rosa,magron2017certified,pldi16-bit-level-fp-verification,zhendong2015,satire}. As expected, none of them can deal with probabilistic inputs.  Affine arithmetic~\cite{affineoriginal} is a well-known technique that extends interval arithmetic~\cite{intervalarithmetic} with information about linear correlations between operands~\cite{affinebook}.
Rigorous affine arithmetic~\cite{rigorousaffine} takes into consideration numerical errors stemming from the use of floating-point arithmetic when computing with affine forms, and it has been implemented in Rosa~\cite{rosa} to perform rigorous worst-case roundoff error analysis.
Our symbolic affine arithmetic (see \cref{sec:symbolicaffine}) used in \Tool evolved from rigorous affine arithmetic by keeping the coefficients of the noise terms symbolic, which often leads to improved precision.
These symbolic terms are very similar to the first-order Taylor approximations
of the roundoff error expressions used in
FPTaylor~\cite{2015_fm_sjrg,solovyev2018rigorous}; our experimental results
also support this since the bounds generated by \Tool for the 100\% confidence
interval are almost always equal to the worst-case bounds computed by FPTaylor.

\section{Conclusions and Future Work}
\label{sec:conclusions}

In this paper, we presented our approach to computing rigorous (i.e., sound)
probabilistic roundoff error bounds for arithmetic floating-point computations
involving random variables. 
First, we showed how to explicitly compute the error distribution of
floating-point arithmetic operations directly from the distributions of the
input random variables.
Then, we demonstrated how the random variable and the corresponding error
distribution are close to being uncorrelated.
We leverage this to be able to rigorously compute operations between random
floating-point variables, and thus perform rigorous probabilistic range
analysis.
We showed how to do this even in the complex case of dependent operands using
our novel combination of p-boxes and SMT solving.
We use our probabilistic range analysis to compute conditional roundoff errors
where, as opposed to the worst-case approach, the (symbolic) error expression
gets maximized constrained to the output landing in a given confidence interval
of interest (e.g., 99\%). 
We implemented our approach in a prototype tool named \Tool, and we compared
\Tool on a popular benchmark suite with state-of-the-art tools for
probabilistic as well as worst-case error analysis.
Our results show that \Tool almost always outperforms the state-of-the-art
probabilistic analysis tool PrAn in terms of generating tighter, but still
rigorous, error bounds.
Moreover, we observe that worst-case roundoff errors can be very pessimistic in
some cases, and that \Tool can reduce such error bounds by several orders of
magnitude.

As future work, we plan to explore the potential use of our probabilistic range
analysis to perform probabilistic overflow detection.
Together with a warning about a potential overflow, a user would also
get from \Tool a rigorous bound on the probability of the overflow happening.
This would allow users to perform a more detailed risk analysis in order to decide
whether a mitigation effort is needed.
We envision a similar use case for the probabilistic division-by-zero
detection.
Finally, in our motivating example (see \cref{sec:overview}), we performed
manual probabilistic precision tuning to further drive the point that
probabilistic error analysis is needed in many settings.
In fact, \Tool can be used as a back-end roundoff error analyzer as a part of
an existing precision tuner (such as FPTuner~\cite{fptuner}).
In such a setup, the confidence interval becomes a new key parameter in the
precision tuning process, thereby allowing for programmers to explore a richer
space of trade-offs.
We plan to explore this line of work in the future.

\newpage

\bibliographystyle{plain}
\bibliography{refs}

\newpage

\appendix

\section{Proofs}
We define $\ceil{x}\defeq \sup\{y\in\R\mid \round(y)=\round(x)\}$ and $\floor{x}\defeq \inf\{z\in\R\mid \round(z)=\round(x)\}$. Recall also that 
\[
\fintvl\stackrel{\triangle}{=}\left\{y\in\R\mid \round(y)=\round(z)\right\}.
\]

\subsection{Proofs for \cref{subsec:LPerror_dist}.}

\begin{myproof}{\cref{thm:LP_errorDensity}}
For $t\in\left]-1,1\right[$, the cumulative distribution function of the measure associated with the random variable $\erel(X)$ is given by:
\begin{align*}
c(t)\stackrel{\triangle}{=}~\Pro{\erel(X)\leq t\uro}=~\Pro{~\bigvee_{z\in\F}\left(\frac{X-z}{X}\leq t\uro\wedge X\in \fintvl[z]\right)}
\end{align*}
As explained in \cref{subsec:LPerror_dist} we exclude the floating-point representable numbers $\{-\infty, 0, \infty\}$ which correspond to the discrete components $\Pro{X\ssvin\fintvl[0]}\delta_1 + \Pro{X\ssvin\fintvl[-\infty]}\delta_{\infty}\ \hspace{-3pt}+ \Pro{X\ssvin\fintvl[\infty]}\delta_{-\infty}$ of the distribution $dist$.

Using the $\sigma$-additivity of measures (see \cref{subsec:prob}) we get the following density for $dist_c$:
\newcommand{\ssin}{\hspace{-2pt}\in\hspace{-2pt}}
\begin{align*}
d(t)=&\dt\sum_{z\in\F\setminus\{-\infty,0,\infty\}}\hspace{-8pt}\Pro{\frac{X-z}{X}\leq t\uro\wedge X\in \fintvl[z]}\\
=&\hspace{-6pt} \sum_{z\in\F^+_{0,-\infty}}\hspace{-8pt}\dt\Pro{\frac{z}{1-t\uro}\hspace{-2pt}\geq\hspace{-2pt} X\wedge X\ssin\fintvl[z]}+\hspace{-4pt}
\sum_{z\in\F^-_{0,\infty}}\hspace{-4pt}\dt\Pro{\frac{z}{1-t\uro}\hspace{-2pt}\leq\hspace{-2pt} X \wedge X\ssin \fintvl[z]}
\end{align*}
where $\F^+_{0,\infty}$ (resp. $\F^-_{0,-\infty}$) denotes the strictly positive (resp. negative) finite floating-point representable numbers.
Since $X$ is described by a probability density function $f:\R\to\R$, we get:
\begin{align}
d(t)
=&\hspace{-6pt}\sum_{z\in\F^+_{0,-\infty}}\hspace{-8pt}\dt\one_{\fintvl[z]}\hspace{-2pt} \left(\frac{z}{1-t\uro}\right)\hspace{-2pt} \int^{\frac{z}{1-t\uro}}_{\floor{z}}\hspace{-4pt}f(s)\hspace{1pt}ds\nonumber+
\hspace{-4pt}\sum_{z\in\F^-_{0,\infty}}\hspace{-6pt} \dt\one_{\fintvl[z]}\hspace{-2pt}\left(\frac{z}{1-t\uro}\right)\hspace{-2pt} \int^{\ceil{z}}_{\frac{z}{1-t\uro}}\hspace{-5pt} f(s)\hspace{1pt}ds\nonumber 
\\
=&\sum_{z\in\F\setminus\{-\infty,0,\infty\}}\hspace{-5pt}\one_{\fintvl[z]}\left(\frac{z}{1-t\uro}\right) f\left(\frac{z}{1-t\uro}\right) \frac{\uro\absv{z}}{(1-t\uro)^2}\nonumber
\end{align}
\end{myproof}

\subsection*{Proofs for \cref{subsec:HPerror_dist}.}

We now give some useful results about floating-point numbers.
The following lemma collects explicit representations of $\ceil{z}$ and $\floor{z}$ which will be useful in what follows, the proof is by direct computation.
\begin{lemma}\label{lem:floorceil}
For $z=z(0,e,k)\in \F$ the values of $\ceil{z}$ and $\floor{z}$ are given by: 
\begin{align*}
\floor{z}& =
\begin{cases}
2^{e-1} & \text{if }e=\emin, k=0\\
2^{e-1}\left(1+\frac{2^{p+1}-1}{2^{p+1}}\right) & \text{if }e>\emin, k=0\\
2^e\left(1+\frac{2k-1}{2^{p+1}}\right) & \text{otherwise}
\end{cases} 
& \qquad & \\
\ceil{z}&=\begin{cases}
z&\text{if }e=n, k=2^p-1\\
2^e\left(1+\frac{2k+1}{2^{p+1}}\right) & \text{otherwise}
\end{cases}
\end{align*}
For $z=z(1,e,k)$ we use the identities $\floor{-z}=-\ceil{z}$ and $\ceil{-z}=-\floor{z}$.
\end{lemma}

Using this lemma we can prove a result relating the length $\tau(z)\defeq\ceil{z}-\floor{z}$ of the interval corresponding to a representable number $z\in\F$ with its absolute value $\absv{z}$. Note that $\ceil{z}-\floor{z}$ is always a positive quantity, hence the relation with $\absv{z}$. The following result will be used heavily in the proof of \cref{thm:HP_errorDensity} . Again, the proof can be obtained by direct computation.

\begin{lemma}\label{lem:Ccoeff}
For any $z(s,e,k)\in\F, z\neq 0$ it is the case that 
\[
u\absv{z} =  C(e,k)\left(\ceil{z}-\floor{z}\right)
\]
where the coefficients $C(e,k)$ are given by:
\begin{align*}
&C(\emin,0) = \frac{2^{p+1}+1}{2^p(2^{p+1}-1)} & & C(e,0)=\frac{2}{3},\quad e>\emin\\
&C(\emax,2^p-1)=\frac{3(2^{p+1}-1)}{2^{p+1}} & & C(e,k)=\frac{2^p+k}{2^{p+1}},\quad\text{otherwise}.
\end{align*}
\end{lemma}
\noindent Finally, we will need the following technical lemma which can also be shown by direct computation.
\begin{lemma}\label{lem:trange}
If $z=z(s,e,k)$ with $e\neq \emin,\emax$ then 
\[
\frac{z}{1-t\uro}\in\fintvl \text{ iff } 
\begin{cases}
\frac{-2^{p+1}}{2^{p+2}-1}\leq t\leq \frac{2^{p+1}}{2^{p+1}+1} & k=0\\
\frac{-2^{p+1}}{2^{p+1}+2k-1}\leq t\leq \frac{2^{p+1}}{2^{p+1}+2k+1}&\text{otherwise}
\end{cases}
\]
\end{lemma}
Note in particular that $\frac{z}{1-t\uro}\in\fintvl$ whenever $\absv{t}\leq \frac{1}{2}$. We are now ready to prove the theorem of \cref{subsec:HPerror_dist}.

\begin{myproof}{\cref{thm:HP_errorDensity}}
We start by plugging the coefficients of Prop. \ref{lem:Ccoeff} into \cref{eq:LPerrorDensity}:
\begin{align*}
d(t)=\sum_{z\in\F\setminus\{-\infty,0,\infty\}}\one_{\fintvl[z]}\left(\frac{z}{1-t\uro}\right) f\left(\frac{z}{1-t\uro}\right) \frac{C(e,k) (\ceil{z}-\floor{z})}{(1-t\uro)^2}.
\end{align*}
We start by considering the case where $\absv{t}\leq \frac{1}{2}$. As shown in \cref{lem:trange} if $\absv{t}\leq \frac{1}{2}$ then $\one_{\fintvl[z]}\left(\frac{z}{1-t\uro}\right) = 1$, and $d(t)$ thus simplifies to
\begin{align*}
d(t)=\sum_{z\in\F\setminus\{-\infty,0,\infty\}} f\left(\frac{z}{1-t\uro}\right) \frac{C(e,k) (\ceil{z}-\floor{z})}{(1-t\uro)^2}.
\end{align*}
We now restrict ourselves to floating-point representable whose exponents are not extremal, \ie $z(s,e,k)$ with $k\notin\{\emin,\emax\}$, and put the usually minuscule contribution made by these numbers in the error term $R$, \ie we will take
 \[
 \Pro{\round(X)=z(s,\emin, k)} +  \Pro{\round(X)=z(s,\emax, k)} 
 \]
to be part of the error term $R$. We do this chiefly for mathematical expediency (to avoid the special behaviour of \cref{lem:Ccoeff}).
We are now left with the sum 
\begin{align}
d(t) =& \frac{1}{(1-t\uro)^2} \sum_{s, e=\emin+1}^{\emax-1} \left(\frac{2}{3}f\left(\frac{z(s,e,0)}{1-t\uro}\right)(\ceil{z(s,e,0)}-\floor{z(s,e,0)})\right. \nonumber\\
&
+\left. \sum_{k=1} \frac{2^p+k}{2^{p+1}} f\left(\frac{z(s,e,k)}{1-t\uro}\right)(\ceil{z(s,e,k)}-\floor{z(s,e,k)}) \right)\nonumber
\\
\stackrel{(1)}{=} & \frac{1}{(1-t\uro)^2} \sum_{s, e=\emin+1}^{\emax-1} \left(\frac{1}{2}f\left(\frac{z(s,e,0)}{1-t\uro}\right)(\ceil{z(s,e,0}-z(1-\uro))\right.\nonumber\\
&
+\left. \sum_{k=1}^{2^{p-1}} \frac{z(s,e,k)}{2^{e+1}} f\left(\frac{z(s,e,k)}{1-t\uro}\right)(\ceil{z(s,e,k)}-\floor{z(s,e,k)}) \right)
\nonumber\\
\stackrel{(2)}{=} & \frac{1}{1-t\uro} \sum_{s, e=\emin+1}^{\emax-1}\frac{1}{2^{e+1}} \left(\frac{z(s,e,0)}{1-t\uro}f\left(\frac{z(s,e,0)}{1-t\uro}\right)(\ceil{z(s,e,0}-z(1-\uro))\right.\nonumber\\
&
+\left. \sum_{k=1}^{2^{p-1}} \frac{z(s,e,k)}{1-t\uro} f\left(\frac{z(s,e,k)}{1-t\uro}\right)(\ceil{z(s,e,k)}-\floor{z(s,e,k)}) \right)
\nonumber\\
\stackrel{(3)}{=} &  \frac{1}{1-t\uro} \sum_{s, e=\emin+1}^{\emax-1}\frac{1}{2^{e+1}} \left(\int_{(-1)^s2^e(1-u)}^{(-1)^s2^e(2-u)} \absv{x} f(x) ~dx \right) + S(e,s,t)\label{eq:case_one_half}
\end{align}
where $(1)$ follows by computing the distance $\ceil{z(s,e,0)}-\floor{z(s,e,0)}$ with \cref{lem:floorceil} and by substituting $k$ for its value in terms of $z$ via $z=2^e(1+\frac{k}{2^p})$, and $(2)$ follows from writing $\frac{1}{2}$ as $\frac{z(s,e,0)}{2^{e+1}}$. The last step follows by noticing that the sum in step (2) approximates the integral in (3) (where the absolute value $\absv{x}$ takes care of the fact that the bounds of the integral will be the 'wrong way round' when $s=1$).  Indeed, when $t=0$ it corresponds to the  so-called \emph{mid-point rule} of numerical integration. We can explicitly compute the error $S(e,s,t)$ incurred from this approximation using standard technique from numerical analysis. Specifically,  for a fixed $z(s,e,k)$ we write the function $xf(x)$ as a Taylor expansion around $\frac{z}{1-tu}f\left(\frac{z}{1-tu}\right)$ with a first-order remainder in Lagrange form to get
\begin{align}
xf(x) = \frac{z}{1-tu}f\left(\frac{z}{1-tu}\right) + \left(x- \frac{z}{1-tu}\right)\left(f'(\xi_x)\xi_x + f(\xi_x)\right)\label{eq:Taylor}
\end{align}
for some $\xi_x$ between $x$ and  $\frac{z}{1-tu}$.  Using this representation we get the error term
\begin{align*}
S(z, t) \defeq & \int_{z_0}^{\ceil{z}} f\left(\frac{z(s,e,k)}{1-t\uro}\right)(\ceil{z(s,e,k)}-\floor{z(s,e,k)})  - \int_{z_0}^{\ceil{z}} xf(x)~dx
\\
& = \left(f'(\xi_k)\xi_k + f(\xi_k)\right) \frac{-zt\uro}{1-t\uro}\left(\ceil{z}-z_0\right)
\end{align*}
where $z_{0}=z(1-\uro)$ when the significand of $z$ is 0 and $z_0=\floor{z}$ otherwise. The last step uses the Mean Value Theorem extract a unique $\xi_k\in[z_0,\ceil{z}]$.. Note that when $t=0$ the error term above disappears. This is completely expected since $t=0$ corresponds to the mid-point rule which is a second-order numerical integration technique (it has no first-order error). We should therefore expand \cref{eq:Taylor} to the second order in the case where $t=0$. however, since we will integrate over every $t$, what happens at $t=0$ can be ignored since the Lebesgue measure is continuous.  Now summing over every significand and using the Mean Value Theorem once more we get and cumulative error
\begin{align*}
S(e,s,t)& \defeq \sum_{k=0}^{2^p-1}S(z(s,e,k), t) \\
&= \left(f'(\xi_{e,s})\xi_{e,s} + f(\xi_e)\right)\sum_{k=0}^{2^p-1}\frac{-z(s,e,k)t\uro}{1-t\uro}\left(\ceil{z(s,e,k)}-z_0\right)
\end{align*}
for some $\xi_{e,s}$ in $[2^e(1-\uro), 2^e(2-\uro)]$ when $s=0$ and  $[-2^e(2-\uro), -2^e(1-\uro)]$ when $s=1$,  and for $z_0$ defined as above. We can bound $\absv{S(e,s,t)}$ as follows. We first compute
\begin{align*}
\absv{\frac{-zt\uro}{1-t\uro}\left(\ceil{z}-z_0\right))}&=\absv{\frac{t}{1-t\uro}} \absv{z}\uro\left(z-z_0\right)\\
&\stackrel{(1)}{=}\absv{\frac{t}{1-t\uro}} \frac{2^p+k}{2^{p+1}}\left(z-z_0\right)^2\\
&\stackrel{(2)}{=}\absv{\frac{t}{1-t\uro}} \frac{2^p+k}{2^{p+1}}\frac{2^{2e}}{2^{2p}}
\end{align*}
where (1) follows by \cref{lem:Ccoeff} and (2) follows from the fact that all the intervals $z-z_0$ have the same size given by $\nicefrac{2^e}{2^p}$. We can now bound $S(e,s,t)$ by
\begin{align}
\absv{S(e,s,t)}&= \absv{f'(\xi_{e,s})\xi_{e,s} + f(\xi_{e,s})}\sum_{k=0}^{2^p-1} \absv{\frac{t}{1-t\uro}} \frac{2^p+k}{2^{p+1}}\frac{2^{2e}}{2^{2p}}\nonumber \\
&=  \absv{f'(\xi_{e,s})\xi_{e,s} + f(\xi_{e,s})}\absv{\frac{t}{1-t\uro}} \frac{3}{4}(2^p-1)\frac{2^{2e}}{2^{2p}}\nonumber \\
&=  \absv{f'(\xi_{e,s})\xi_{e,s} + f(\xi_{e,s})}\absv{\frac{t}{1-t\uro}} \frac{3}{4}(1-u)\frac{2^{2e}}{2^{p}}\nonumber \\
& \leq \absv{f'(\xi_{e,s})\xi_{e,s} + f(\xi_{e,s})} \frac{3}{4}\frac{2^{2e}}{2^{p}}\label{eq:error_term}
\end{align}
since $\absv{t}\leq 1$.

We now consider the case where $\frac{1}{2}<\absv{t}\leq 1$.  We can follow the derivation leading to  \cref{eq:case_one_half} by simply appending a coefficient given by $\one_{\fintvl[z]}$ at every step up to the penultimate one.  Now using \cref{lem:trange} and substituting $k$ by its value in terms of $z$ via $z=2^e(1+\frac{k}{2^p})$ we get:
\[
\one_{\fintvl[z]}\left(\frac{z(s,e,k)}{1-t\uro}\right) = 1\Leftrightarrow
\begin{cases}
 z\leq 2^e(\frac{1}{t}-\uro) & z> 0\\
 z\leq 2^e(-\frac{1}{t}+\uro) & z <0.
\end{cases}
\]
This means that in the last step of the derivation of  \cref{eq:case_one_half} we simply need to replace the bounds of the integrals by $2^e(\frac{1}{t}-\uro)$ and $2^e(\frac{1}{t}+\uro)$, which leads to the final expression
\[
\frac{1}{(1-t\uro)}\sum\limits_{s, e=\emin+1}^{\emax-1} \int_{(-1)^s2^e(1-\uro)}^{(-1)^s2^e(\frac{1}{\absv{t}}-\uro)}\frac{\absv{x}}{2^{e+1}} f(x) ~dx 
\]

In terms of error, it is clear that since we're approximating the same integrals $\int xf(x)~dx$ in the same way, but on a more narrow interval, the error term \cref{eq:error_term} also provides an upper bound the error made in the integrals $\int_{(-1)^s2^e(1-\uro)}^{(-1)^s2^e(\frac{1}{\absv{t}}-\uro)}\frac{\absv{x}}{2^{e+1}} f(x) ~dx$,  in other words 
\[
\absv{S(e,s,t)}\leq  \absv{f'(\xi_e)\xi_e + f(\xi_e)} \frac{3}{4}\frac{2^{2e}}{2^{p}}
\]
also when $\nicefrac{1}{2}\leq \absv{t}\leq 1$. It thus follows by the triangular inequality that we can bound the total error made from converting discrete sums into integrals by
\[
\frac{3}{4} \left(\sum_{s, \emin<e<\emax} \left(f'(\xi_{e,s})\xi_{e,s} + f(\xi_{e,s})\right)\frac{2^{2e}}{2^{p}}\right)
\]
as claimed.
\end{myproof}

\subsection*{Proofs for \cref{subsec:typical}.}

\begin{myproof}{\cref{thm:typical_dist}}
We start with the exact expression given in \cref{eq:LPerrorDensity}, namely
\begin{align*}
d(t)=\sum_{z\in\F\setminus\{-\infty,0,\infty\}}\one_{\fintvl[z]}\left(\frac{z}{1-t\uro}\right) f\left(\frac{z}{1-t\uro}\right) \frac{\uro\absv{z}}{(1-t\uro)^2}
\end{align*}
We are now going to sum over significands (since they are assumed to be equiprobable). Using \cref{lem:Ccoeff} to re-write $u\absv{z}$ and ignoring the terms with maximal exponents as we did in the proof of \cref{thm:HP_errorDensity} (this is not strictly necessary, but it make the proof a lot less cumbersome) we get:
\begin{align}
d(t) =&\sum_{k=0}^{2^p-1}\frac{C(e,k)}{(1-tu)^2}\sum_{e=emin+1}^{emax-1}\sum_{s=0}^1 f\left(\frac{z(e,s,k)}{1-tu}\right)(\ceil{z(e,s,k)}-\floor{z(e,s,k)})\nonumber
\\
=&\sum_{k=0}^{2^p-1}\frac{C(e,k)}{(1-tu)^2}\Pro{\round(X)=z(s,e,k)})+ S(k,p, t)\nonumber 
\\
&=\frac{1}{2^p(1-tu)^2}\left(\frac{2}{3}+\sum_{k=1}^{2^p-1}\frac{2^p+k}{2^{p+1}}\right)+ S(k,p,t) \nonumber
\\
&=\frac{1}{2^p(1-tu)^2}\left(\frac{2}{3}+\frac{3(2^{p}-1)}{4}\right)+ S(k,p,t)\label{eq:pdfleq1/2}
\end{align}
where $S(k,p,t)$  is an error term quantifying the error made by the numerical integration scheme evaluating $\Pro{\round(X)=z(s,e,k)})$ as the sum in the first step of the derivation above. It can be quantified in the same way as in the proof of \cref{thm:HP_errorDensity}. Crucially $\lim_{p\to\infty}S(k,p,t)=0$ for every $k,t$. 
We thus get that for $\absv{t}<\nicefrac{1}{2}$
\[
\lim_{p\to\infty}d(t) = \lim_{p\to\infty} \frac{1}{2^p(1-t2^{-p})^2}\left(\frac{2}{3}+\frac{3(2^{p}-1)}{4}\right)+ S(k,p) = \frac{3}{4}.
\]

As shown by \cref{lem:trange}, when $\frac{1}{2}< t\leq 1$  not all mantissas in the sum of \cref{eq:LPerrorDensity} are possible,. Using the explicit characterisation of \cref{lem:trange}, we get that for $\nicefrac{1}{2}< \absv{t}\leq 1$ \eqref{eq:pdfleq1/2} becomes 
\begin{align*}
d(t)&=\frac{1}{2^p(1-tu)^2}\left(\frac{2}{3}+\sum_{k=1}^{\alpha(t)}\frac{2^p+k}{2^{p+1}}\right) + S(k,p,t)
\\
&=\frac{1}{2^p(1-tu)^2}\left(\frac{2}{3}+\frac{1}{2}\alpha(t)+\frac{1}{2^{p+2}}(\alpha(t)^2)\right)+ S(k,p,t)
\end{align*}
where $\alpha(t)=\floor{2^p(\frac{1}{t}-1)-\frac{1}{2}}$ is the usual floor function applied to $2^p(\frac{1}{t}-1)-\frac{1}{2}$.  Taking the limit of the equation above we get
\begin{align}
\lim_{p\to\infty}d(t) & = \lim_{p\to\infty}  \frac{1}{(1-t2^{-p})^2}\left(\frac{1}{2}\left(\frac{1}{t}-1\right)+\frac{1}{4}\left(\frac{1}{t}-1\right)^2\right)+S(k,p,t)\nonumber \\
&= \frac{1}{2}\left(\frac{1}{t}-1\right)+\frac{1}{4}\left(\frac{1}{t}-1\right)^2\label{eq:pdfgeq1/2}
\end{align}
The case where $-1\leq t< -\frac{1}{2}$ is treated in the same way and yields the same asymptotic distribution.  Combining \ref{eq:pdfleq1/2} and \ref{eq:pdfgeq1/2} we get
\begin{align*}
\lim_{p\to\infty}d(t) =\begin{cases}
\frac{3}{4}&\text{if }t\in\left[-\frac{1}{2} ,\frac{1}{2}\right]  \\
\frac{1}{2}\left(\frac{1}{t}-1\right)+\frac{1}{4}\left(\frac{1}{t}-1\right)^2& \text{else }\\
\end{cases}
\end{align*}
as announced.
\end{myproof}

\subsection*{Proofs for \cref{subsec:covar}.}

We assume throughout this section that $\erel(X)$ is distributed according to $d_{hp}$ in \cref{eq:HPerrorDensity}.
We start by bounding $\Exp{\erel(X)}$. The proof of the following Proposition is a little technical but not conceptually difficult and can be found in the Appendix.

\begin{proposition}\label{prop:ExpError}
$K\frac{\uro}{6}\leq \Exp{\erel(X)}\leq K\frac{4\uro}{3}$
where
\[
K=\sum_{s,e=\emin+1}^{\emax-1} \int_{(-1)^s2^e(1-\uro)}^{(-1)^s2^e(2-u)} \frac{\absv{x}}{2^{e+1}} f(x) ~dx
\]
\end{proposition}
\begin{proof}
Consider the definition \cref{eq:HPerrorDensity}, when $\frac{1}{2}< \absv{t}\leq 1$. Note that the terms
\[
K(t)\defeq \sum\limits_{s, e=\emin+1}^{\emax-1} \int_{(-1)^s2^e(1-\uro)}^{(-1)^s2^e(\frac{1}{\absv{t}}-\uro)} \frac{\absv{x}}{2^{e+1}} f(x) ~dx
\]
are symmetric around 0 due to the dependence on $\absv{t}$ (note also that the constant $K= K(\frac{1}{2})$). It follows by a simple change of variable that
\begin{align*} 
\int_{-1}^{-\frac{1}{2}} \frac{t}{(1-t\uro)^2}K(t)~dt &= \int_{1}^{\frac{1}{2}}\frac{-s}{(1+s\uro)^2}K(-s)~(-ds)\\
&= \int_{\frac{1}{2}}^1 \frac{-s}{(1+s\uro)^2}K(s)~ds
\end{align*}
and thus
\begin{align*}
\int_{-1}^{-\frac{1}{2}} d(t)t~dt + \int^{1}_{\frac{1}{2}} d(t)t~dt& = -\int_{\frac{1}{2}}^{1}\frac{t}{(1+t\uro)^2} K(t)~dt +  \int^{1}_{\frac{1}{2}} d(t)t~dt
\\
& \geq - \int_{\frac{1}{2}}^{1}\frac{t}{(1-t\uro)^2} K(t)~dt + \int^{1}_{\frac{1}{2}} d(t)t~dt = 0
\end{align*}
From this it follows that the net contribution $\Exp{\erel(X)}$ of the two `wings' of the distribution is positive and in particular that
\[
\int_{-\frac{1}{2}}^{\frac{1}{2}}d(t) t~dt\leq \int_{-1}^{1}d(t) t~dt=\Exp{\erel(X)}.
\]
Moreover, by definition of $K$ we have
\[
\int_{-\frac{1}{2}}^{\frac{1}{2}}d(t) t~dt = K\int_{-\frac{1}{2}}^{\frac{1}{2}}\frac{t}{(1-t\uro)^2}~dt
\]
We now integrate by parts to get
\begin{align*}
K\int_{-\frac{1}{2}}^{\frac{1}{2}}\frac{t}{(1-t\uro)^2}~dt &= K \left.\frac{t}{\uro(1-t\uro)}\right|^{\frac{1}{2}}_{-\frac{1}{2}} - \int_{-\frac{1}{2}}^{\frac{1}{2}} \frac{1}{\uro(1-t\uro)}~dt
\\
&=\frac{K}{\uro}\left.\left(\frac{t}{\uro(1-t\uro)}-\frac{\ln(1-t\uro)}{\uro}\right)\right|^{\frac{1}{2}}_{-\frac{1}{2}}
\\
&=\frac{K}{\uro}\left(\frac{1}{1-\frac{\uro^2}{4}}- \frac{2}{\uro}\arctanh\left(\frac{\uro}{2}\right) \right)
\end{align*} 
For small values of $\uro$ the term above is very nearly linear with slope $\frac{1}{6}$ and intercept 0. To an extremely good approximation we therefore get for small values of $\uro$ (say below $2^{-2}$) that
\[
K\int_{-\frac{1}{2}}^{\frac{1}{2}}\frac{t}{(1-t\uro)^2}~dt = K\frac{\uro}{6}
\]
which proves the lower bound.

For the upper bound we proceed in a similar way. Note that 
\begin{align*}
\Exp{\erel(X)} =  &\int_{-1}^{-\frac{1}{2}}d(t) t~dt + \int_{-\frac{1}{2}}^{\frac{1}{2}}d(t) t~dt +\int_{\frac{1}{2}}^1d(t) t~dt
\\
= &
~\int_{-1}^{-\frac{1}{2}}\frac{t}{(1-t\uro)^2}K(t) ~dt +
K \int_{-\frac{1}{2}}^{\frac{1}{2}}\frac{t}{(1-t\uro)^2} ~dt + \int_{\frac{1}{2}}^1\frac{t}{(1-t\uro)^2} K(t) ~dt
\\
\leq & ~K \int_{-1}^{1}\frac{t}{(1-t\uro)^2}~dt\qquad\qquad\qquad\text{(since }K(t)\leq K\text{ )}
\\
= & ~ \frac{K}{\uro} \left(\frac{2}{(1-\uro^2)} - \frac{2}{\uro} \arctanh(\uro) \right)
\end{align*}
Again, the expression above is very nearly linear in $\uro$ for small values of $\uro$ (say below $2^{-2}$), with slope $\frac{4}{3}$ and intercept 0. It follows that for small values of $\uro$
\[
\Exp{\erel(X)}\leq K\frac{4\uro}{3}
\]
proving the upper bound.
\flushright{$\square$}
\end{proof}
Note that $K\leq 1$ since $\absv{x}\leq 2^{e+1}$ in each summand. The \emph{expected relative error $\Exp{\erel(X)}$ will therefore always be at most of the order of $\uro$} by the previous result. This makes good intuitive sense.

We now turn to $\Exp{X.\erel(X)}$. The proof of the following result is conceptually more involved. It is based on the fundamental definition of Lebesgue integration as the limit of a monotone sequence of so-called simple functions. 

\begin{proposition}\label{prop:expXerelX}$\Exp{X.\erel(X)} = \sum_{s,e} f((-1)^s2^e)(-1)^s2^{2e}\frac{3\uro^2}{2}$
\end{proposition}
In particular $\Exp{X.\erel(X)}=0$ if $f$ is symmetric. More generally, the value of $\Exp{X.\erel(X)}$ is determined by the competing terms $2^{2e}$ and $\uro^2$. A quick calculation shows that as long as the distribution $f$ assigns most of its mass to values below $2^{\frac{p}{2}}$, $\Exp{X.\erel(X)}$ will be of order $u$. In fact, for all the benchmarks considered in \cref{sec:evaluation} bar one, $\Exp{X.\erel(X)}$ is of the order of $\uro$ or smaller.
\begin{proof}
In order to (over)approximate the support of the joint distribution $\mathbb{P}$ and the integral \cref{eq:expdef}, we consider for each $n\in\N$ the following collection of sets
\begin{align*}
A^n(z,k) \defeq &\left[\floor{z}+\frac{k}{n}\tau(z), \floor{z}+\frac{k+1}{n}\tau(z)\right]\times \\
&\left[\erel\left(\floor{z}+\frac{k}{n}\tau(z)\right), \erel\left(\floor{z}+\frac{k+1}{n}\tau(z)\right)\right]
\end{align*}
where $z\in \F, 0\leq k < n$, and $\tau(z)\defeq \ceil{z}-\floor{z}$ is the size of the interval $\fintvl$. These sets contain pairs $(x,\erel(x))$ where $x$ is restricted to a sub-interval of $\fintvl$ whose size is controlled by $n$. For any $n$, the support of $\mathbb{P}$ is included in the union $\bigcup_{z,k}A^n(z,k)$, by definition of $\mathbb{P}$. Moreover,  by definition of the joint probability $\mathbb{P}$ and the assumption that $f$ is constant on $\fintvl$, we have
\begin{align*}
\Pro{A^n(z,k)} = & \int_{\floor{z}+\frac{k}{n}\tau(z)}^{\floor{z}+\frac{k+1}{n}} f(x)~dx\\
= & f(z) \frac{\tau(z)}{n}
\end{align*}
The sets $A^n(k,z)$ were chosen is such a way that the \emph{simple functions} (see \cite{dudley2002real}) $h_n: \R\times [-1,1]\to\R$ given by
\[
h_n(x,t) = \sum_{z,k} \sup\{y\uro s : (y,s)\in A^n(z,k)\} .\one_{A^n(z,k)}(x,t)
\]
provide a monotonically decreasing (over)approximation of the function
\[
h(x,t) = x\uro t
\]
which we want to integrate. Using the common notation of integration theory: $h_n \downarrow h$. It now follows from the definition of the Lebesgue integral that
\begin{align*}
\Exp{X.\erel(X)} & = \int_{\R\times [-1,1]} x\uro t~d\mathbb{P}\\
& \defeq \int_{\R\times [-1,1]} h(x,t)~d\mathbb{P}\\
&\leq \int_{\R\times [-1,1]} h_n(x,t)~d\mathbb{P}\\
&= \sum_{z,k} \sup_{(x,t)\in A^n(z,k)} x\uro t~ \Pro{A^n(z,k)}
\end{align*}
The supremum is easily seen to be
\[
\sup_{(x,t)\in A^n(z,k)} x\uro t = 
\begin{cases}
\left(\floor{z}+\frac{k+1}{n}\tau(z)\right)\erel(\floor{z}+\frac{k+1}{n}\tau(z)) & z > 0
\\
\left(\floor{z}+\frac{k}{n}\tau(z)\right)\erel(\floor{z}+\frac{k}{n}\tau(z)) & z<0
\end{cases}
\]
Clearly, multiplying any real by its own relative error will yield its absolute error:
\[
\sup_{(x,t)\in A^n(z,k)} x\uro t = 
\begin{cases}
(\floor{z}+\frac{k+1}{n}\tau(z))-z & z > 0\\
(\floor{z}+\frac{k}{n}\tau(z))-z & z<0
\end{cases}
\]
and we thus have
\begin{align*}
\Exp{X.\erel(X)}\leq & \sum_{k,z>0} (\floor{z}+\frac{k+1}{n}\tau(z))-z) f(z) \frac{\tau(z)}{n} + \\
&\sum_{k,z<0} (\floor{z}+\frac{k}{n}\tau(z))-z) f(z) \frac{\tau(z)}{n}\\
 = &\sum_{z>0} f(z) \sum_k (\floor{z}+\frac{k+1}{n}\tau(z))-z)\frac{\tau(z)}{n}+\\
&\sum_{z<0} f(z) \sum_k (\floor{z}+\frac{k}{n}\tau(z))-z)\frac{\tau(z)}{n}.
\end{align*}
It is not too hard to see that as $n\to\infty$ the expression $\sum_k (\floor{z}+\frac{k+1}{n}\tau(z))-z)\frac{\tau(z)}{n}$ and $\sum_k (\floor{z}+\frac{k}{n}\tau(z))-z)\frac{\tau(z)}{n}$ both approach the Riemann integral
\[
\int_{\floor{z}}^{\ceil{z}} (x-z)~dx
\] 
In other words:
\[
\Exp{X.\erel(X)}\leq \sum_z f(z) \int_{\floor{z}}^{\ceil{z}} (x-z)~dx
\]
There are now two cases. If the significand of $z$ is zero then the interval $\fintvl$ is not symmetric around $z$, it is in fact divided in a $\frac{1}{3}$:$\frac{2}{3}$ ratio around $z$ and we therefore have
\begin{align*}
\int_{\floor{z}}^{\ceil{z}} (x-z)~dx &= \int_{\floor{z}}^{z+\frac{\tau(z)}{3}} (x-z)~dx+ \int_{z+\frac{\tau(z)}{3}}^{\ceil{z}} (x-z)~dx\\
&=0+\int_{z+\frac{\tau(z)}{3}}^{\ceil{z}} (x-z)~dx\\
\end{align*} 
For $z=z(0,e,0)=2^e$ the expression above yields:
\[
\int_{2^e(1+\uro)}^{2^e(1+2\uro)} (x-2^e)~dx = 2^{2e}\frac{3\uro^2}{2}.
\]
Similarly for $z=z(1,e,0)=-2^e$ we get
\[
\int^{-2^e(1+\uro)}_{-2^e(1+2\uro)} (x+2^e)~dx = -2^{2e}\frac{3\uro^2}{2}.
\]
All other intervals $\fintvl$ are centred around $z$ and therefore
\[
\int_{\floor{z}}^{\ceil{z}} (x-z)~dx = 0
\]
We thus have the following bound for the expectation $\Exp{X.\erel(X)}$:
\[
\Exp{X.\erel(X)}\leq \sum_{s,e} f((-1)^s2^e)(-1)^s2^{2e}\frac{3\uro^2}{2}
\]
A completely analogous argument using an under-approximating sequence of simple function $h_n\uparrow h$ defined by $\inf(A^n(z,k))$ instead of $\sup(A^n(z,k))$ shows that
\[
\Exp{X.\erel(X)}\geq \sum_{s,e} f((-1)^s2^e)(-1)^s2^{2e}\frac{3\uro^2}{2}
\]
and the equality follows.
\flushright{$\square$}
\end{proof}

\begin{myproof}{\cref{thm:covar}}
This is a direct corollary of \cref{prop:ExpError} and \cref{prop:expXerelX}.
\end{myproof}

\section{Implementation}\label{sec:appdx-implementation}

\subsubsection{Independent Operation.}\label{sec:appdx-independent}
We report the pseudo-code in \cref{independentalg}.
When the operands $X,Y$ are independent, we can compute $Z=X\iop Y$ with a simple pairwise operation (line 3, 4) between focal elements using interval arithmetic (line 5), and we multiply the corresponding probabilities (line 6).  
No matter the independence, we still populate and propagate the trace for future dependent operations (line 8).
This is very much similar to the pen-and-paper approach using a joint table.
\begin{algorithm}[H]
	\caption{Independent Operation $Z=X\iop Y$}\label{independentalg}
	\begin{algorithmic}[1]
		\Function {ind\_op}{$DS_X,\iop,DS_Y,trace_x,trace_y$} 
		\State $DS_Z=list()$
		\ForAll {$ ([x_1, x_2], p_{x}) \in DS_X$} 
		\ForAll {$ ([y_1, y_2], p_{y}) \in DS_Y$}
		\State $[z_1, z_2]=[x_1, x_2]\;op\;[y_1, y_2]$ \Comment{operation between intervals}
		\State $p_{z}=p_{x}*p_{y}$
		\State $DS_Z.append(([z_1, z_2],p_z))$
		\EndFor
		\EndFor
		\State $trace_Z=trace_X \cup trace_Y\cup\{Z=X\iop Y\}$
		\State \textbf{return} $DS_Z$, $trace_Z$
		\EndFunction
	\end{algorithmic}
\end{algorithm}

\subsection{Linear Programming Routine}\label{sec:appdx-lp}
We can create a linear programming (LP) routine to delineate upper bound and lower bounds for the p-box.
In the following we report the pseudo-code for the evaluation of the upper bound of a p-box. The lower bound is similar.

\begin{align*}
&\text{\textbf{maximize}}\sum_{evp\:\in\:z_{i,j}} p_{z_{i,j}}\\
&\text{\textbf{subject to}}\quad 0\leq p_{z_{i,j}}\leq 1\\
&\qquad\qquad\quad\quad \forall i \in [1,N], \sum_{1\leq j\leq N} p_{z_{i,j}} = p_{y_{i}} \\
&\qquad\qquad\quad\quad \forall j \in [1,N], \sum_{1\leq i\leq N} p_{z_{i,j}} = p_{x_{i}}\\
\end{align*}
The probabilities in the DS structures of the operands ($p_{y_i}$, $p_{x_i}$) are called the \emph{marginals}. The focal elements we use to populate $DS_Z$ are called the \emph{insiders} ($p_{z_{i,j}}$).
The constraints in the LP program force the insiders to have a probability between 0 and 1, and they relate the insiders with the marginals. The insiders have to sum up to the marginals. This is very similar to a so called \emph{joint table}.

The LP program takes in input an \emph{evaluation point} (evp) and returns a probability value.
In order to construct the $DS_Z$ exactly we should pick one evaluation point per focal element. This has quadratic complexity.
A good trade-off between accuracy and execution time consists in picking only $N$ out of the $N^2$ evaluation points (e.g. using some heuristics), at the price of a slightly over-approximation.
We run the linear programming routines in parallel, because the analysis of a single evaluation point is completely independent from the others.

\section{Experimental Evaluation}
\label{sec:appdx-evaluation}

\begin{table*}[tb]
	\scriptsize
	\caption{Comparison between \Tool, given uniform (uni), normal (norm) and exponential (exp) input distributions, and FPTaylor. The PAF columns report 99\% of the support of the output range distribution. The FPTaylor columns report the worst-case output ranges. The asterisk (*) highlights a difference of more than one order of magnitude between \Tool and FPTaylor.}
    \label{rangeanalysis}
	\renewcommand{\arraystretch}{1.1}
	\setlength{\tabcolsep}{0.1em} 
	\centering
	\begin{tabular}{@{\extracolsep{2.pt}}p{1.4cm}L{2.55cm}L{2.7cm}L{2.6cm}L{2.5cm}@{}}
		\toprule
		\multirow{1}{*}{Benchmark} &  \multicolumn{1}{c}{\Tool (uni)} &
		\multicolumn{1}{c}{\Tool (norm)} & \multicolumn{1}{c}{\Tool (exp)} &
		\multicolumn{1}{c}{FPTaylor} \\
		\midrule
		bspline0 & [-5.85e-08, 1.67e-01] & [-5.85e-08, 1.67e-01] & [1.22e-01, 1.67e-01]* & [-5.72e-08, 1.67e-01]\\
		\mydashline{}
		bspline1 & [1.65e-01, 6.67e-01] & [1.65e-01, 6.67e-01] & [6.47e-01, 6.67e-01] & [1.55e-01, 6.75e-01]\\
		\mydashline{}
		bspline2 & [1.67e-01, 6.67e-01] & [1.67e-01, 6.68e-01] & [1.67e-01, 2.34e-01] & [1.59e-01, 6.78e-01] \\
		\mydashline{}
		bspline3 & [-1.67e-01, 5.22e-08] & [-1.67e-01, 5.22e-08] & [-4.04e-04, 5.22e-08]* & [-1.67e-01, 4.22e-08]\\
		\mydashline{}
		classids0 & [-2.32e01, 2.47e01] & [-2.32e01, 7.22e00] & [-1.43e01, -8.93e-01]* & [-2.32e01, 2.47e01]\\
		\mydashline{}
		classids1 & [-1.02e01, 1.20e01] & [-3.81e00, 1.21e01] & [-2.80e-01, 5.18e00]* & [-1.02e01, 1.21e01]\\
		\mydashline{}
		classids2 & [-1.67e01, 1.33e01] & [-9.80e00, 1.33e01] & [-8.06e-01, 8.91e00]* & [-1.68e01, 1.34e01]\\
		\mydashline{}
		doppler1 & [-1.43e02, -3.39e-02] & [-7.56e-02, -3.39e-02]* & [-7.81e00, -3.39e-02]* & [-1.39e02, -3.16e-02]\\
		\mydashline{}
		doppler2 & [-2.45e02, -2.27e-02] & [-9.29e00, -2.27e-02]* & [-1.07e01, -2.27e-02]* & [-2.32e02, -2.08e-02] \\
		\mydashline{}
		doppler3 & [-8.62e-01, -5.05e-01] & [-7.98e00, -5.05e-01]* & [-7.39e00, -5.06e-01]* & [-8.35e01, -5.03e-01]\\
		\mydashline{}
		filter1 & [-1.40e00, 1.40e00] & [-1.40e00, 1.40e00] & [-1.12e-01, 1.12e-01]* & [-1.40e00, 1.40e00]\\
		\mydashline{}
		filter2 & [-2.04e00, 2.04e00] & [-2.04e00, -2.04e00] & [-9.54e-01, 8.81e-01]* & [-2.06e00, 2.06e00]\\
		\mydashline{}
		filter3 & [-2.36e00, 2.36e00] & [-2.36e00, 2.36e00] & [-2.36e00, 2.36e00] & [-2.39e00, 2.39e00] \\
		\mydashline{}
		filter4 & [-3.26e00, 3.26e00] & [-3.26e00, 3.26e00] & [-3.27e00, 3.27e00] & [-3.32e00, 3.32e00] \\
		\mydashline{}
		rigidbody1 & [-7.05e02, 7.05e02] & [-5.05e01, 5.69e01]* & [-4.45e01, 6.00e01]* & [-7.05e02, 7.05e02]\\
		\mydashline{}
		rigidbody2 & [-5.60e04, 5.87e04] & [-4.77e03, 8.70e03]* & [-1.44e02, 3.74e02]* & [-5.63e04, 5.87e04] \\
		\mydashline{}
		sine & [-1.00e00, 1.00e00] & [-1.00e00, 1.00e00] & [-2.61e-01, 4.21e-01] & [-1.02e00, 1.02e00]\\
		\mydashline{}
		solvecubic & [-9.27e-01, 4.81e01] & [-9.27e-01, 1.37e01] & [-9.27e-01, 9.87e00] & [-9.34e-01, 4.82e01]\\
		\mydashline{}
		sqrt & [-3.35e02, 1.48e00] & [-7.57e-01, 1.48e00]* & [-2.11e-01, 1.46e00]* & [-3.38e02, 1.50e00] \\
		\mydashline{}
		traincars1 & [-2.67e03, 5.44e03] & [2.03e03, 5.44e03]* & [3.36e03, 5.44e03]* & [-2.67e03, 5.44e03] \\
		\mydashline{}
		traincars2 & [-3.65e03, 3.85e03] & [-1.39e03, 1.45e03] & [-6.41e02, 6.57e02] & [-3.65e03, 3.85e03] \\
		\mydashline{}
		traincars3 & [-6.70e04, 6.73e04] & [-2.52e04, 2.43e04] & [-1.42e04, 1.33e04] & [-6.70e04, 6.73e04] \\
		\mydashline{}
		traincars4 & [-5.84e05, 5.69e05] & [-2.33e05, 5.69e05] & [-1.15e05, 1.27e05] & [-5.84e05, 5.69e05] \\
		\mydashline{}
		trid1 & [-1.28e04, 1.26e04] & [-5.55e03, 7.72e03] & [-6.91e02, 7.68e02] & [-1.29e04, 1.27e04] \\
		\mydashline{}
		trid2 & [-3.26e04, 1.33e04] & [-1.65e03, 1.59e03]* & [-1.96e03, 1.81e03]* & [-3.27e04, 1.35e04]\\
		\mydashline{}
		trid3 & [-5.28e04, 1.42e04] & [-2.98e03, 3.95e03]* & [-4.06e03, 4.27e03]* & [-5.30e04, 1.40e04] \\
		\mydashline{}
		trid4 & [-7.26e04, 1.92e04] & [-6.10e03, 1.17e04]* & [-8.13e03  1.05e04]* & [-7.28e04, 1.54e04]\\
		\bottomrule
	\end{tabular}
\end{table*}

%
%
%
%

\end{document}